\pdfoutput=1

\let\TeXyear\year
\documentclass{ieeeaccess}

\let\year\TeXyear

\usepackage{tcolorbox}
\NewSpotColorSpace{PANTONE}
\AddSpotColor{PANTONE}{PANTONE3015C}{PANTONE\SpotSpace 3015\SpotSpace C}{1 0.3 0 0.2}
\SetPageColorSpace{PANTONE}
\definecolor{accessblue}{cmyk}{1,0.3,0,0.2}
\definecolor{greycolor}{cmyk}{0,0,0,.8}

\usepackage{cite}
\usepackage{amsmath,amssymb,amsfonts}
\usepackage[compatibility=false]{caption}
\usepackage{subcaption}
\usepackage{amsthm}
\usepackage{algorithmic}
\usepackage[ruled,vlined]{algorithm2e}

\usepackage{graphicx}
\usepackage{textcomp}
\usepackage[dvipsnames]{xcolor}
\usepackage{booktabs}
\usepackage{makecell}
\usepackage{float}
\usepackage{url}
\usepackage[utf8]{inputenc}

 \usetikzlibrary{patterns}
\usepackage{comment}
\usepackage{soul}
\usepackage{bm}
\usepackage{siunitx}

\usepackage{tikz}
\usetikzlibrary{arrows.meta,positioning,calc,fit,backgrounds}
\usetikzlibrary{arrows.meta}
\usepackage{pgfplots}
\usepackage{pifont}
\pgfplotsset{compat=1.17}

\theoremstyle{definition}
\newtheorem{definition}{Definition}

\theoremstyle{plain}
\newtheorem{theorem}{Theorem}
\newtheorem{lemma}{Lemma}

\usepgfplotslibrary{groupplots}
\tikzstyle{blackdot}=[circle,fill=black,minimum size=2mm,inner sep=0pt]
\tikzstyle{whitedot}=[circle,draw=black,fill=black,minimum size=2mm,inner sep=0pt]

\tikzset{
    solidEdge/.style={thick},
    dottedEdge/.style={thick,dashed}
}

\def\BibTeX{{\rm B\kern-.05em{\sc i\kern-.025em b}%
\kern-.08em T\kern-.1667em\lower.7ex\hbox{E}\kern-.125emX}}

\begin{document}
\history{Date of publication xxxx 00, 0000, date of current version xxxx 00, 0000.}
\doi{10.1109/ACCESS.2024.0429000}

\title{Distributed Balanced Butterfly Counting in Signed Bipartite Graphs}
\author{%
\uppercase{Kiran Mekala}\authorrefmark{1},
\uppercase{Apurba Das}\authorrefmark{2},
\uppercase{Suman Banerjee}\authorrefmark{3},
}

\address[1]{Computer Science and Information Systems, BITS Pilani Hyderabad Campus, Hyderabad, India (e-mail: p20220017@hyderabad.bits-pilani.ac.in)}

\address[2]{Computer Science and Information Systems, BITS Pilani Hyderabad Campus, Hyderabad, India (e-mail: apurba@hyderabad.bits-pilani.ac.in)}

\address[3]{Department of Computer Science and Engineering, Indian Institute of Technology Jammu, Jammu \& Kashmir, India (e-mail: suman.banerjee@iitjammu.ac.in)}


\markboth
{Mekala \headeretal: D-BBC: Distributed Balanced Butterfly Counting in Signed Bipartite Graphs}
{Mekala \headeretal: D-BBC: Distributed Balanced Butterfly Counting in Signed Bipartite Graphs}

\corresp{Corresponding author: MEKALA KIRAN (e-mail: p20220017@hyderabad.bits-pilani.ac.in).}

\begin{abstract}
The balanced butterfly is a fundamental primitive for analyzing signed bipartite graphs and provides a basis for studying higher-order structural properties, such as clustering coefficients and community structure. Despite its importance, existing approaches primarily rely on serial algorithms for balanced butterfly counting, which become inefficient on large-scale graphs. To address this limitation, we propose a distributed algorithm, \texttt{D-BBC}, based on a hybrid MPI+TBB framework that exploits MPI for inter-process communication and Intel TBB for intra-node parallelism. We conduct an experimental assessment of the proposed approach across 15 real-world datasets. Experimental results demonstrate that, on a single-node distributed system, \texttt{D-BBC} achieves average speedups of 1321$\times$ and 16.2$\times$ over the serial \texttt{BB2K} and multi-core \texttt{M-BBC} implementations, respectively. Furthermore, \texttt{D-BBC} achieves a maximum speedup of 23.58$\times$ over the distributed baseline \texttt{S-Monarch} in terms of end-to-end execution time. These results demonstrate the efficiency of the proposed distributed approach and its potential to enable high-performance signed motif analysis on large-scale bipartite graphs.
\end{abstract}


\titlepgskip=-21pt

\maketitle

\section{Introduction}
\label{sec:introduction}
Signed bipartite graphs are ubiquitous. They model interactions between two disjoint sets of entities, where relationships are associated with positive or negative semantics, such as trust/distrust, like/dislike, and approval/disapproval~\cite{chen2021maximum,torres2016drug}. For example, in user-product networks, users express positive or negative feedback through ratings or reviews, while in legislator-bill networks, legislators cast yay or nay votes on bills, naturally yielding polarized relationships~\cite{derr2019balance}. Unlike traditional bipartite graphs, which assume homogeneous relationships and treat all edges identically~\cite{sanei2018butterfly,wang2019vertex,wang2022accelerated}, signed bipartite graphs explicitly capture the polarity of interactions and therefore provide a richer representation of many real-world systems~\cite{kiran2024efficient,sun2022maximal}.
 Figure~\ref{fig:gray_vs_signed} shows an example of unsigned and signed bipartite networks constructed from a user-product network.

 The rapid growth of large-scale signed bipartite graphs has created an increasing demand for efficient graph analytics. Among various graph mining tasks, counting and enumerating network motifs are fundamental because they constitute the basic building blocks of complex networks~\cite{tang2024monarch,wang2019vertex}. Numerous cohesive structures have been studied in bipartite graphs, including bicliques, bicores, and bitrusses~\cite{heider1946attitudes,B2,B3,chen2022efficient,B5,luo2023efficient}. Among them, the butterfly (i.e., a complete $2\times2$ biclique) is the smallest and most fundamental motif, serving as the basis for applications such as graph decomposition, community detection, dense subgraph mining, and link prediction~\cite{wang2014rectangle,sanei2018butterfly,wang2022accelerated,shi2022parallel,B31}. To extend this concept to signed bipartite graphs, Derr \emph{et al.}~\cite{derr2019balance} introduced the notion of a \emph{balanced butterfly} based on social balance theory~\cite{heider1946attitudes}. Balanced butterflies have become an important primitive for analyzing structural balance in signed bipartite networks and have found applications in sign prediction, multidrug discovery, and higher-order signed network analysis~\cite{chen2021maximum,sun2022maximal,kiran2024efficient}.

\begin{figure}[t]
    \centering
    \begin{tikzpicture}[scale=0.9]

        \tikzset{
            img/.style={
                inner sep=0pt,
                anchor=center
            }
        }

        \node[img] at (-0.5,1.4) (u0) {\includegraphics[width=0.65cm]{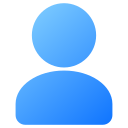}};
        \node[left=-0.15cm of u0] {$u_1$};

        \node[img] at (-0.5,0.4) (u1) {\includegraphics[width=0.65cm]{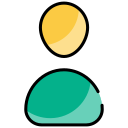}};
        \node[left=-0.15cm of u1] {$u_2$};

        \node[img] at (-0.5,-0.6) (u2) {\includegraphics[width=0.65cm]{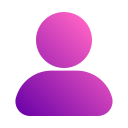}};
        \node[left=-0.15cm of u2] {$u_3$};

        \node[img] at (-0.5,-1.6) (u3) {\includegraphics[width=0.65cm]{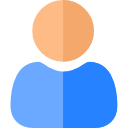}};
        \node[left=-0.15cm of u3] {$u_4$};

        \node[img] at (2,1.4) (v0) {\includegraphics[width=0.65cm]{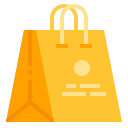}};
        \node[right=-0.15cm of v0] {$p_1$};

        \node[img] at (2,0.4) (v1) {\includegraphics[width=0.65cm]{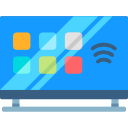}};
        \node[right=-0.15cm of v1] {$p_2$};

        \node[img] at (2,-0.6) (v2) {\includegraphics[width=0.65cm]{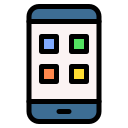}};
        \node[right=-0.15cm of v2] {$p_3$};

        \node[img] at (2,-1.6) (v3) {\includegraphics[width=0.65cm]{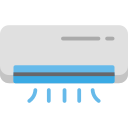}};
        \node[right=-0.15cm of v3] {$p_4$};

        \foreach \i in {0,1,2,3} {
            \foreach \j in {0,1,2,3} {
                \draw[line width=1pt, color=gray] (u\i.east) -- (v\j.west);
            }
        }

        \node[img] at (4.5,1.4) (bu0) {\includegraphics[width=0.65cm]{ICA3PP_2026/Images/u1.png}};
        \node[left=-0.15cm of bu0] {$u_1$};

        \node[img] at (4.5,0.4) (bu1) {\includegraphics[width=0.65cm]{ICA3PP_2026/Images/u2.png}};
        \node[left=-0.15cm of bu1] {$u_2$};

        \node[img] at (4.5,-0.6) (bu2) {\includegraphics[width=0.65cm]{ICA3PP_2026/Images/u3.png}};
        \node[left=-0.15cm of bu2] {$u_3$};

        \node[img] at (4.5,-1.6) (bu3) {\includegraphics[width=0.65cm]{ICA3PP_2026/Images/u4.png}};
        \node[left=-0.15cm of bu3] {$u_4$};

        \node[img] at (7,1.4) (bv0) {\includegraphics[width=0.65cm]{ICA3PP_2026/Images/p1.png}};
        \node[right=-0.15cm of bv0] {$p_1$};

        \node[img] at (7,0.4) (bv1) {\includegraphics[width=0.65cm]{ICA3PP_2026/Images/p2.png}};
        \node[right=-0.15cm of bv1] {$p_2$};

        \node[img] at (7,-0.6) (bv2) {\includegraphics[width=0.65cm]{ICA3PP_2026/Images/p3.png}};
        \node[right=-0.15cm of bv2] {$p_3$};

        \node[img] at (7,-1.6) (bv3) {\includegraphics[width=0.65cm]{ICA3PP_2026/Images/p4.png}};
        \node[right=-0.15cm of bv3] {$p_4$};

        \draw[thick,black] (bu0.east) -- (bv0.west);
        \draw[thick,black] (bu0.east) -- (bv1.west);
        \draw[thick,black] (bu0.east) -- (bv2.west);
        \draw[thick,dashed,red] (bu0.east) -- (bv3.west);

        \draw[thick,black] (bu1.east) -- (bv0.west);
        \draw[thick,black] (bu1.east) -- (bv1.west);
        \draw[thick,black] (bu1.east) -- (bv2.west);
        \draw[thick,black] (bu1.east) -- (bv3.west);

        \draw[thick,dashed,red] (bu2.east) -- (bv0.west);
        \draw[thick,black] (bu2.east) -- (bv1.west);
        \draw[thick,black] (bu2.east) -- (bv2.west);
        \draw[thick,black] (bu2.east) -- (bv3.west);

        \draw[thick,dashed,red] (bu3.east) -- (bv0.west);
        \draw[thick,dashed,red] (bu3.east) -- (bv1.west);
        \draw[thick,black] (bu3.east) -- (bv2.west);
        \draw[thick,black] (bu3.east) -- (bv3.west);

        \node at (1, -2.2) {\textbf{(a)}};
        \node at (1, -2.7) {\textbf{(Bipartite graph)}};
        \node at (6, -2.2) {\textbf{(b)}};
        \node at (6, -2.7) {\textbf{(Signed bipartite graph)}};

    \end{tikzpicture}

    \caption{Example of a user-product network: an unsigned bipartite graph (grey edges) and a signed bipartite graph (solid/dashed edges represent positive/negative interactions).}
    \label{fig:gray_vs_signed}
\end{figure}

Despite their importance, counting balanced butterflies in large-scale signed bipartite graphs remains computationally challenging. Real-world graphs are typically massive and sparse, with highly skewed degree distributions, leading to irregular workloads and substantial computational and memory requirements. Existing studies on balanced butterfly counting have primarily focused on sequential algorithms~\cite{derr2019balance,chen2021maximum,kiran2024efficient}, while parallel and distributed efforts have largely targeted butterfly counting in unsigned bipartite graphs~\cite{shi2022parallel,xia2024gpu,wang2024parallelization,tang2024monarch}. However, these techniques cannot be applied directly to signed bipartite graphs because edge signs impose additional constraints during butterfly verification. In addition, distributed execution also requires efficient graph partitioning, communication minimization, and workload balancing. Consequently, scalable distributed algorithms for balanced butterfly counting in signed bipartite graphs remain largely unexplored.

\subsection{Motivation and Challenges}

The increasing scale of signed bipartite graphs makes exact balanced butterfly counting computationally challenging, particularly for vertices with large neighborhoods. Existing serial approaches~\cite{derr2019balance,kiran2024efficient,chung2023maximum} and shared-memory approaches~\cite{kiran2026multi} are constrained by the computational and memory resources of a single machine, limiting their applicability to large-scale graphs. Although distributed algorithms have been developed for butterfly counting in \emph{unsigned} bipartite graphs, they do not consider edge signs and therefore cannot directly support balanced butterfly counting in signed bipartite graphs~\cite{tang2024monarch,weng2022distributed}. This motivates the development of a distributed solution that can exploit the computational resources of multiple processing nodes while preserving the sign information required for balanced butterfly verification.

Extending distributed butterfly counting to signed bipartite graphs introduces several challenges. Since a butterfly may span multiple graph partitions, local counting can miss cross-partition butterflies or result in duplicate counting when replicated neighborhoods are used. In addition, edge signs must be preserved during partitioning and communication to correctly identify balanced butterflies. The highly irregular distribution of neighborhood sizes can further lead to workload imbalance and straggler processes. Therefore, an effective distributed algorithm must address three key challenges: (i) ensuring that every balanced butterfly is counted exactly once, (ii) minimizing inter-rank communication, and (iii) balancing the computational workload across MPI ranks. To address these challenges, we propose \texttt{D-BBC}, a hybrid MPI+TBB algorithm that combines distributed-memory parallelism with intra-node shared-memory parallelism. \texttt{D-BBC} employs workload-balanced pivot assignment, communication-efficient neighborhood exchange, and distributed reduction to achieve exact balanced butterfly counting while reducing communication overhead and improving computational scalability.

The main contributions of this paper are summarized as follows.

\begin{itemize}

\item We adapt and extend the Monarch algorithm~\cite{tang2024monarch}, originally developed for counting butterflies in unsigned bipartite graphs, to the signed setting. The resulting algorithm, \texttt{S-Monarch}, enables exact counting of balanced butterflies in signed bipartite graphs.

\item We extend our prior CPU-based algorithm, \texttt{BB2K}~\cite{kiran2024efficient}, by designing and implementing a vertex-level parallel algorithm, \texttt{M-BBC}, that accelerates wedge-based balanced butterfly counting while avoiding the enumeration of unbalanced substructures.

\item We propose \texttt{D-BBC}, a distributed algorithm for exact balanced butterfly counting in signed bipartite graphs. We implement \texttt{D-BBC} using a hybrid MPI+TBB framework, enabling distributed-memory execution with intra-node parallelism.

\item We conduct extensive experiments on 15 real-world signed bipartite graphs and compare \texttt{D-BBC} against serial (\texttt{BB2K}), multi-core (\texttt{M-BBC}), and distributed (\texttt{S-Monarch}) baselines, demonstrating substantial speedups across the evaluated datasets.

\end{itemize}

The rest of the paper is organized as follows. Section~\ref{Sec:related} reviews the related work. Section~\ref{Sec:prob_stmt} presents the preliminaries and problem formulation. Section~\ref{Sec:alg} describes the proposed distributed algorithm. Section~\ref{Sec:EE} presents the experimental evaluation of the proposed methods. Finally, Section~\ref{Sec:Con} concludes the paper.

\subsection{Applications of balanced butterflies}

 We list a few applications to motivate balanced butterfly counting.
 \begin{itemize}

    \item \textbf{Detection of multidrug combinations~\cite{torres2016drug}.} The relationship between drugs and targets can be classified as positive (activator) or negative (inhibitor). This sign information helps understand drug-target interactions in multi-drug combinations. As illustrated in Fig.~\ref{fig:Drug}, balanced butterflies are useful for detecting underlying patterns where multiple drugs exhibit similar effects on common targets, enabling the identification of synergistic effects or anomalies. Panel~(a) shows examples where drug actions are coherent at each target and the cycle is balanced, whereas panel~(b) shows incoherent drug actions on each target. Such patterns help identify synergistic interactions or contradictions, supporting the design of effective therapeutic strategies.

\begin{figure}[htbp]
\centering
\begin{tikzpicture}[node distance=0.2cm and 0.2cm,
  every node/.style={inner sep=0pt},
  edge/.style={thick}, 
  scale=0.80, transform shape
]

\newcommand{\drawbutterfly}[7]{
  \begin{scope}[shift={(#1,#2)}]
    \node (d#3a) at (0, 0) {\includegraphics[width=0.7cm]{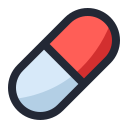}};
    \node (d#3b) at (0, -1.0) {\includegraphics[width=0.7cm]{Images/drug1.png}};
    \node[below=1pt of d#3a] {\scriptsize Drug 1};
    \node[below=1pt of d#3b] {\scriptsize Drug 2};

    \node (t#3a) at (1.6, 0) {\includegraphics[width=0.7cm]{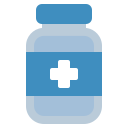}};
    \node (t#3b) at (1.6, -1.0) {\includegraphics[width=0.7cm]{Images/target.png}};
    \node[below=1pt of t#3a] {\scriptsize Target 1};
    \node[below=1pt of t#3b] {\scriptsize Target 2};

    \draw[edge, #4] (d#3a.center) -- (t#3a.center);
    \draw[edge, #5] (d#3a.center) -- (t#3b.center);
    \draw[edge, #6] (d#3b.center) -- (t#3a.center);
    \draw[edge, #7] (d#3b.center) -- (t#3b.center);
  \end{scope}
}

\drawbutterfly{0}{0}{1}{solid}{solid}{solid}{solid}
\node at (0.8, -2.1) {};

\drawbutterfly{3}{0}{2}{solid}{dashed}{solid}{dashed}
\node at (3.8, -2.1) {\textbf{(a)}};

\drawbutterfly{6}{0}{3}{dashed}{dashed}{dashed}{dashed}
\node at (6.8, -2.1) {};

\drawbutterfly{3}{-3.1}{4}{dashed}{solid}{solid}{dashed}
\node at (3.8, -5.2) {\textbf{(b)}};

\end{tikzpicture}
\caption{Coherent/incoherent butterflies in drug-target signed networks.}
\label{fig:Drug}
\end{figure}
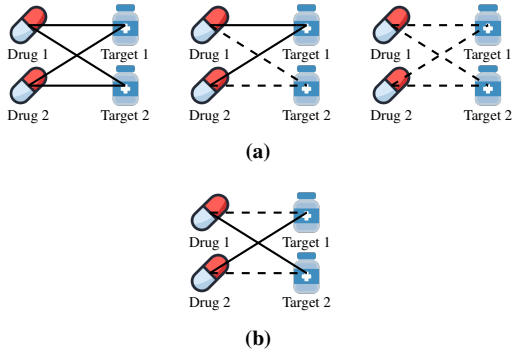

\item \textbf{Detecting dense subgraphs~\cite{sariyuce2018}:} Balanced butterfly counts can reveal dense regions in signed bipartite graphs using both vertex- and edge-based measures. For instance, in  Fig.~\ref{fig:comparison}(a), vertices $a, b, e,$ and $f$ each participate in two balanced butterflies, while $c$ and $d$ are part of three. Since their induced subgraph contains only one butterfly, forming a \texttt{2-tip dense region}. Similarly, in  Fig.~\ref{fig:comparison}(b), the central edges $(c,3), (c,4), (d,3), (d,4)$ each belong to a balanced butterfly, giving them a wing number of 1 and collectively defining a \texttt{1-wing}. Additionally, two $(3, 2)$-bicliques $(abc, 12)$ and $(def, 56)$ have edges participating in two balanced butterflies (\texttt{2-wing}), highlighting highly connected substructures. Vertex $g$ has no balanced butterflies.
\end{itemize}

\begin{figure}[ht]
\centering

\begin{tikzpicture}[scale=0.7]

\tikzstyle{blackdot} = [circle, draw=black, fill=red!40, minimum size=0.22cm, inner sep=0pt]
\tikzstyle{whitedot} = [circle, draw=black, fill=blue!40, minimum size=0.2cm, inner sep=0pt]

\node[blackdot, label=above:{1}] (1) at (0,1.5) {};
\node[blackdot, label=above:{2}] (2) at (1,1.5) {};
\node[blackdot, label=above:{3}] (3) at (2,1.5) {};
\node[blackdot, label=above:{4}] (4) at (3,1.5) {};
\node[blackdot, label=above:{5}] (5) at (4,1.5) {};
\node[blackdot, label=above:{6}] (6) at (5,1.5) {};

\node[whitedot, label=below:{a}] (a) at (0,0) {};
\node[whitedot, label=below:{b}] (b) at (1,0) {};
\node[whitedot, label=below:{c}] (c) at (2,0) {};
\node[whitedot, label=below:{d}] (d) at (3,0) {};
\node[whitedot, label=below:{e}] (e) at (4,0) {};
\node[whitedot, label=below:{f}] (f) at (5,0) {};
\node[whitedot, label=below:{g}] (g) at (6,0) {};

\draw[dashed] (a)--(1);
\draw[thick] (a)--(2);
\draw[dashed] (b)--(1);
\draw[thick] (b)--(2);
\draw[dashed] (c)--(1);
\draw[thick] (c)--(2);
\draw[dashed] (c)--(3);
\draw[dashed] (c)--(4);
\draw[dashed] (d)--(3);
\draw[dashed] (d)--(4);
\draw[thick] (d)--(5);
\draw[dashed] (d)--(6);
\draw[thick] (e)--(5);
\draw[dashed] (e)--(6);
\draw[thick] (f)--(5);
\draw[dashed] (f)--(6);
\draw[thick] (g)--(6);

\begin{scope}[on background layer]
\node[draw, dashed, fill=green!20, fill opacity=0.35,
      rounded corners,
      fit=(a)(b)(c)(d)(e)(f)(1)(2)(3)(4)(5)(6),
      inner sep=12pt] {};
\end{scope}

\node at (3,-2.2) {\textbf{(a): Tip-based hierarchical dense region.}};

\end{tikzpicture}

\vspace{0.4cm}

\begin{tikzpicture}[scale=0.75]

\tikzstyle{blackdot} = [circle, draw=black, fill=blue!40, minimum size=0.22cm, inner sep=0pt]
\tikzstyle{whitedot} = [circle, draw=black, fill=red!40, minimum size=0.2cm, inner sep=0pt]

\node[whitedot, label=above:{1}] (1) at (0,1.5) {};
\node[whitedot, label=above:{2}] (2) at (1,1.5) {};
\node[whitedot, label=above:{3}] (3) at (2,1.5) {};
\node[whitedot, label=above:{4}] (4) at (3,1.5) {};
\node[whitedot, label=above:{5}] (5) at (4,1.5) {};
\node[whitedot, label=above:{6}] (6) at (5,1.5) {};

\node[blackdot, label=below:{a}] (a) at (0,0) {};
\node[blackdot, label=below:{b}] (b) at (1,0) {};
\node[blackdot, label=below:{c}] (c) at (2,0) {};
\node[blackdot, label=below:{d}] (d) at (3,0) {};
\node[blackdot, label=below:{e}] (e) at (4,0) {};
\node[blackdot, label=below:{f}] (f) at (5,0) {};
\node[blackdot, label=below:{g}] (g) at (6,0) {};

\draw[dashed] (a)--(1);
\draw[thick] (a)--(2);
\draw[dashed] (b)--(1);
\draw[thick] (b)--(2);
\draw[dashed] (c)--(3);
\draw[dashed] (c)--(4);
\draw[dashed] (d)--(3);
\draw[dashed] (d)--(4);
\draw[thick] (e)--(5);
\draw[dashed] (e)--(6);
\draw[thick] (f)--(5);
\draw[dashed] (f)--(6);
\draw[dashed] (c)--(1);
\draw[thick] (c)--(2);
\draw[thick] (d)--(5);
\draw[dashed] (d)--(6);
\draw[thick] (g)--(6);

\begin{scope}[on background layer]
\draw[dashed, thick, fill=blue!20, fill opacity=0.35]
plot[smooth cycle, tension=0.7] coordinates {
    (-0.2,-0.4) (2.2,-0.4) (2.2,0.4) (1.3,1.9) (-0.2,1.9) (-0.6,0.4)
};

\node[draw, dashed, rounded corners, fill=green!20, fill opacity=0.35,
      fit=(c)(d)(3)(4), inner sep=4pt] {};

\draw[dashed, thick, fill=red!20, fill opacity=0.35]
plot[smooth cycle, tension=0.9] coordinates {
    (3.2,-0.4) (5.4,-0.4) (5.6,0.2) (5.6,1.9) (3.8,1.9) (2.6,0.4)
};
\end{scope}

\node at (3,-2.2) {\textbf{(b): Edge-based dense regions}};

\end{tikzpicture}

\caption{Comparison of vertex-centric (tip-based) and edge-centric (wing-based) dense region detection using balanced butterflies.}
\label{fig:comparison}
\end{figure}

\section{Related Work}
\label{Sec:related}
\subsection{Butterfly Counting in Bipartite Graphs}

Butterfly counting has been extensively studied in unsigned bipartite graphs due to its importance in graph mining, community detection, dense subgraph discovery, and network analysis. Early studies focused on exact butterfly counting through edge- and wedge-based enumeration strategies~\cite{wang2014rectangle,sanei2018butterfly,zhu2018fast}, significantly reducing redundant computations during butterfly enumeration. As the size of bipartite graphs increased, subsequent research improved scalability by developing shared-memory parallel algorithms~\cite{wang2019vertex,wang2024parallelization}, GPU implementations~\cite{xia2024gpu}, I/O-efficient techniques~\cite{shi2022parallel}, and algorithms for temporal~\cite{cai2023efficient,papadias2024counting}, uncertain~\cite{zhou2021butterfly}, and streaming graphs~\cite{meng2026counting}.

To process graphs that exceed the memory capacity of a single machine, several distributed butterfly counting algorithms have also been proposed. Early distributed approaches employed graph partitioning and parallel execution across multiple machines~\cite{wang2014rectangle,weng2022distributed}. More recently, Tang \emph{et al.}~\cite{tang2024monarch} proposed Monarch, a communication-efficient distributed framework that combines neighborhood expansion with workload balancing to achieve scalable exact butterfly counting in large unsigned bipartite graphs.

Despite these advances, all existing distributed butterfly counting algorithms are designed exclusively for unsigned bipartite graphs. They neither preserve edge-sign information nor verify structural balance during enumeration. Consequently, these approaches cannot be directly extended to balanced butterfly counting in signed bipartite graphs, where every candidate butterfly must additionally satisfy balance constraints while avoiding duplicate counting and excessive communication across distributed graph partitions.

\subsection{Balanced Structures in Signed Bipartite Graphs}

Structural balance theory has recently been extended from signed graphs to signed bipartite graphs, giving rise to several balanced cohesive structures. Derr \emph{et al.}~\cite{derr2019balance} introduced the balanced butterfly as the fundamental balanced motif for signed bipartite graphs and demonstrated its usefulness in characterizing structural balance. Subsequently, Sun \emph{et al.}~\cite{sun2022maximal} investigated maximal balanced signed bicliques, while Chung \emph{et al.}~\cite{chung2023maximum} proposed the balanced $(k,\epsilon)$-bitruss model for discovering cohesive balanced subgraphs. More recently, Kiran \emph{et al.}~\cite{kiran2024efficient} generalized balanced butterfly analysis by proposing an efficient algorithm for balanced $(2,k)$-biclique counting. In addition, recent work has developed shared-memory multi-core and GPU algorithms for exact balanced butterfly counting, significantly improving the performance of single-machine implementations~\cite{kiran2026multi}.

Although these studies have substantially advanced balanced butterfly analysis in signed bipartite graphs, they are all limited to single-machine execution. To the best of our knowledge, no existing work has investigated exact distributed balanced butterfly counting in signed bipartite graphs. This paper addresses this gap by proposing D-BBC, the first distributed algorithm for exact balanced butterfly counting using a hybrid MPI+TBB framework.

 \section{Preliminaries and Problem Definition}
\label{Sec:prob_stmt}

We consider a signed bipartite graph $G = (U, V, E = E^{+}\cup E^{-})$, in which $U$ and $V$ are the bipartitions, and $E \subseteq U \times V$ is the edge set partitioned into positive edges $E^+$ and negative edges $E^-$. We define $\text{sign}(e)= ``+"$ for an edge $e \in E^{+}$ and $\text{sign}(e)= ``-"$ for an edge $e \in E^{-}$. For a vertex $u\in U$, let $d(u)$ refer to the degree of $u$ and $\Gamma(u)$ indicate the neighbors of $u$ in $G$. Now, we define some basic terms for our problem.

\begin{definition}[\textbf{Butterfly}~\cite{wang2019vertex}]
Given a bipartite graph $G=(U,V,E)$, a butterfly is a cycle of length four induced by vertices $(u_i, u_j, v_i, v_j)$, where $u_i, u_j \in U$ and $v_i, v_j \in V$, such that all four possible edges between $\{u_i,u_j\}$ and $\{v_i,v_j\}$ exist in $G$.

\label{def:butterfly}
\end{definition}

\begin{definition}[\textbf{Balanced Butterfly}~\cite{derr2019balance}]
A \textbf{butterfly} in $G$ is a subgraph induced by four distinct vertices $\{u, u'\} \subseteq U$ and $\{v, v'\} \subseteq V$ such that all four edges $(u,v)$, $(u,v')$, $(u',v)$, and $(u',v')$ exist in $E$. Let $k$ denote the number of negative edges (i.e., edges with sign $0$) in the butterfly. The butterfly is said to be \textbf{balanced} if $k$ is even, otherwise \textbf{unbalanced}. Fig.~\ref{figure:Fig_Two} represents the balanced and unbalanced butterflies.

\begin{figure}[ht]
\centering
\begin{tikzpicture}[scale=0.9, line width=0.3pt]

    \tikzset{
      blackdot/.style={circle, draw=black, fill=blue!40, minimum size=6pt, inner sep=0pt},
      whitedot/.style={circle, draw=black, fill=red!40, minimum size=6pt, inner sep=0pt},
      edge/.style={line width=0.8pt, shorten >=1pt, shorten <=1pt}
    }

\draw[dotted, thick] (-0.3,1.3) rectangle (8.5,-1.4);


\node[blackdot,label=below:$u_k$] (a_u2) at (0,0) {};
\node[blackdot,label=above:$u_i$] (a_u1) at (0,0.7) {};
\node[whitedot,label=below:$v_\ell$] (a_v2) at (0.7,0) {};
\node[whitedot,label=above:$v_j$] (a_v1) at (0.7,0.7) {};
\draw[thick] (a_u2)--(a_v2);
\draw[thick] (a_u2)--(a_v1);
\draw[thick] (a_u1)--(a_v2);
\draw[thick] (a_u1)--(a_v1);
\node at (0.35,-0.75) {(a)};

\node[blackdot,label=below:$u_k$] (b_u2) at (1.2,0) {};
\node[blackdot,label=above:$u_i$] (b_u1) at (1.2,0.7) {};
\node[whitedot,label=below:$v_\ell$] (b_v2) at (1.9,0) {};
\node[whitedot,label=above:$v_j$] (b_v1) at (1.9,0.7) {};
\draw[thick] (b_u2)--(b_v2);
\draw[red,thick] (b_u2)--(b_v1);
\draw[red,thick] (b_u1)--(b_v2);
\draw[thick] (b_u1)--(b_v1);
\node at (1.55,-0.75) {(b)};

\node[blackdot,label=below:$u_k$] (c_u2) at (2.4,0) {};
\node[blackdot,label=above:$u_i$] (c_u1) at (2.4,0.7) {};
\node[whitedot,label=below:$v_\ell$] (c_v2) at (3.1,0) {};
\node[whitedot,label=above:$v_j$] (c_v1) at (3.1,0.7) {};
\draw[red,thick] (c_u2)--(c_v2);
\draw[red,thick] (c_u2)--(c_v1);
\draw[thick] (c_u1)--(c_v2);
\draw[thick] (c_u1)--(c_v1);
\node at (2.75,-0.75) {(c)};

\node[blackdot,label=below:$u_k$] (d_u2) at (3.6,0) {};
\node[blackdot,label=above:$u_i$] (d_u1) at (3.6,0.7) {};
\node[whitedot,label=below:$v_\ell$] (d_v2) at (4.3,0) {};
\node[whitedot,label=above:$v_j$] (d_v1) at (4.3,0.7) {};
\draw[red,thick] (d_u2)--(d_v2);
\draw[thick] (d_u2)--(d_v1);
\draw[red,thick] (d_u1)--(d_v2);
\draw[thick] (d_u1)--(d_v1);
\node at (3.95,-0.75) {(d)};

\node[blackdot,label=below:$u_k$] (e_u2) at (4.8,0) {};
\node[blackdot,label=above:$u_i$] (e_u1) at (4.8,0.7) {};
\node[whitedot,label=below:$v_\ell$] (e_v2) at (5.5,0) {};
\node[whitedot,label=above:$v_j$] (e_v1) at (5.5,0.7) {};
\draw[red,thick] (e_u2)--(e_v2);
\draw[red,thick] (e_u2)--(e_v1);
\draw[red,thick] (e_u1)--(e_v2);
\draw[red,thick] (e_u1)--(e_v1);
\node at (5.15,-0.75) {(e)};

\draw[dotted,thick] (5.8,1.3)--(5.8,-1.4);


\node[blackdot,label=below:$u_k$] (f_u2) at (6.2,0) {};
\node[blackdot,label=above:$u_i$] (f_u1) at (6.2,0.7) {};
\node[whitedot,label=below:$v_\ell$] (f_v2) at (6.9,0) {};
\node[whitedot,label=above:$v_j$] (f_v1) at (6.9,0.7) {};
\draw[red,thick] (f_u2)--(f_v2);
\draw[thick] (f_u2)--(f_v1);
\draw[thick] (f_u1)--(f_v2);
\draw[thick] (f_u1)--(f_v1);
\node at (6.55,-0.75) {(f)};

\node[blackdot,label=below:$u_k$] (g_u2) at (7.4,0) {};
\node[blackdot,label=above:$u_i$] (g_u1) at (7.4,0.7) {};
\node[whitedot,label=below:$v_\ell$] (g_v2) at (8.1,0) {};
\node[whitedot,label=above:$v_j$] (g_v1) at (8.1,0.7) {};
\draw[red,thick] (g_u2)--(g_v2);
\draw[red,thick] (g_u2)--(g_v1);
\draw[red,thick] (g_u1)--(g_v2);
\draw[thick] (g_u1)--(g_v1);
\node at (7.75,-0.75) {(g)};

\node at (2.7,-1.15) {\textbf{Balanced}};
\node at (7.0,-1.15) {\textbf{Unbalanced}};

\end{tikzpicture}
\caption{Illustration of balanced and unbalanced butterflies.}
\label{figure:Fig_Two}
\end{figure}
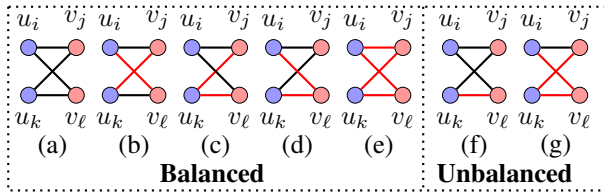

\end{definition}

\begin{definition}[\textbf{Vertex Priority}~\cite{wang2019vertex}]
Let $a,b \in \{U \cup V\}$ be two vertices. We say that vertex $a$ has higher priority than vertex $b$, denoted by $p(a) > p(b)$, if one of the following conditions holds:
\begin{enumerate}
    \item $degree(a) > degree(b)$;
   \item $id(a)>id(b)$, if $degree(a)=degree(b)$
\end{enumerate}
Here, $id(a)$ is the vertex ID of $a$.
\label{def-priority}
\end{definition}

\begin{definition}[\textbf{Wedge ($\lor$)}~\cite{wang2019vertex}]
Given a bipartite graph $G=(U,V,E)$, let $u_i,u_k \in U$ and $v_j \in V$. 
A wedge $\lor(u_i,v_j,u_k)$ is a path of length two that starts at $u_i$, passes through $v_j$, and ends at $u_k$, as illustrated in Fig.~\ref{fig:wedges}(a). 
We assume an ordering on vertices in $U$ such that $p(u_i) > p(u_k)$.
\label{def:wedge}
\end{definition}

\begin{definition}[\textbf{Symmetric Wedge ($\lor^s$)}~\cite{kiran2024efficient}]
A wedge $\lor(u_i,v_j,u_k)$ in a signed bipartite graph $G=(U,V,E)$ is called a \textbf{symmetric wedge} if the two edges $(u_i,v_j)$ and $(u_k,v_j)$ have the same sign, i.e., both are positive or both are negative. Such a wedge is denoted by $\lor^s(u_i,v_j,u_k)$ and is illustrated in Fig.~\ref{fig:wedges}(b).
\label{def:wedgesym}
\end{definition}

\begin{definition}[\textbf{Asymmetric Wedge ($\lor^a$)}~\cite{kiran2024efficient}]
A wedge $\lor(u_i,v_j,u_k)$ in a signed bipartite graph $G=(U,V,E)$ is called an \textbf{asymmetric wedge} if the two edges $(u_i,v_j)$ and $(u_k,v_j)$ have different signs, i.e., one is positive and the other is negative. Such a wedge is denoted by $\lor^a(u_i,v_j,u_k)$ and is illustrated in Fig.~\ref{fig:wedges}(c).
\label{def:wedgeasym}
\end{definition}

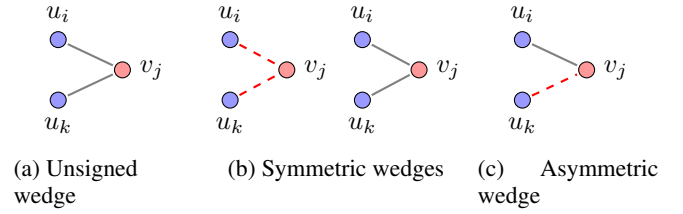
\begin{figure}[ht]
\centering

\tikzset{
  leftvertex/.style={
    circle,
    draw=black,
    fill=blue!40,
    minimum size=6pt,
    inner sep=0pt,
    outer sep=0pt
  },
  rightvertex/.style={
    circle,
    draw=black,
    fill=red!40,
    minimum size=6pt,
    inner sep=0pt,
    outer sep=0pt
  },
  edge/.style={
    line width=0.8pt,
    shorten >=1pt,
    shorten <=1pt
  },
  signededge/.style={
    edge,
    dashed,
    red
  }
}

\begin{subfigure}[t]{0.28\linewidth}
\centering
\begin{tikzpicture}

\node[leftvertex,label=above:{$u_i$}] (u1) at (0,0.4) {};
\node[leftvertex,label=below:{$u_k$}] (u2) at (0,-0.4) {};
\node[rightvertex,label=right:{$v_j$}] (v1) at (0.85,0) {};

\draw[edge,gray] (u1) -- (v1);
\draw[edge,gray] (u2) -- (v1);

\end{tikzpicture}
\caption{Unsigned \\ wedge}
\end{subfigure}
\hfill
\begin{subfigure}[t]{0.40\linewidth}
\centering
\begin{tikzpicture}

\node[leftvertex,label=above:{$u_i$}] (a1) at (0,0.4) {};
\node[leftvertex,label=below:{$u_k$}] (a2) at (0,-0.4) {};
\node[rightvertex,label=right:{$v_j$}] (v1) at (0.75,0) {};

\draw[signededge] (a1) -- (v1);
\draw[signededge] (a2) -- (v1);

\node[leftvertex,label=above:{$u_i$}] (b1) at (1.75,0.4) {};
\node[leftvertex,label=below:{$u_k$}] (b2) at (1.75,-0.4) {};
\node[rightvertex,label=right:{$v_j$}] (v2) at (2.50,0) {};

\draw[edge,gray] (b1) -- (v2);
\draw[edge,gray] (b2) -- (v2);

\end{tikzpicture}
\caption{Symmetric wedges}
\end{subfigure}
\hfill
\begin{subfigure}[t]{0.28\linewidth}
\centering
\begin{tikzpicture}

\node[leftvertex,label=above:{$u_i$}] (u1) at (0,0.4) {};
\node[leftvertex,label=below:{$u_k$}] (u2) at (0,-0.4) {};
\node[rightvertex,label=right:{$v_j$}] (v1) at (0.85,0) {};

\draw[edge,gray] (u1) -- (v1);
\draw[signededge] (u2) -- (v1);

\end{tikzpicture}
\caption{Asymmetric wedge}
\end{subfigure}

\caption{Wedge structures in unsigned and signed bipartite graphs. 
Gray solid edges represent positive edges, while red dashed edges represent negative edges.}
\label{fig:wedges}

\end{figure}


We further propose a lemma to establish the
significance of these wedges.

\begin{lemma}
Any balanced butterfly in a signed bipartite graph can be formed either using a pair of symmetric wedges or a pair of asymmetric wedges, but not both.
\label{lemma:lem-1}
\end{lemma}
\begin{proof}
Considering a signed bipartite graph $ G = (U, V, E)$, let $butterfly({u,w,v,x})$ represent a balanced butterfly in $G$ in which $u,w \in U$ and $v,x \in V$ with two wedges $\lor(u,v,w)$ and $\lor(u,x,w)$.

\noindent\textbf{case-1.}~$\lor^{s}(u,v,w)$ and $\lor^{s}(u,x,w)$. In this case, the number of negative edges is $0$, $2$, or $4$, thus even.

\noindent\textbf{case-2.}~$\lor^{a}(u,v,w)$ and $\lor^{a}(u,x,w)$. In this case, the number of negative edges is $2$, thus even.

\noindent\textbf{case-3.}~$\lor^{s}(u,v,w)$ and $\lor^{a}(u,x,w)$. In this case, the number of negative edges is $1$ or $3$, thus odd. So, this combination cannot make a balanced butterfly.

This completes the proof.
\end{proof} 

\noindent\textbf{Problem Statement.}~Given a signed bipartite graph $G=(U, V, E)$, the problem is defined as determining the total number of balanced butterflies in $G$.

\section{algorithms}
\label{Sec:alg}

We compare \texttt{D-BBC} against two baselines: a state-of-the-art serial
algorithm and a distributed algorithm adapted from prior work on unsigned
graphs.

\subsubsection{Serial Baseline: \texttt{BB2K}}
As research on balanced butterfly counting in signed bipartite graphs remains limited, we adopt our prior serial algorithm ~\cite{kiran2024efficient}, denoted \textbf{\texttt{BB2K}} (where $k=2$), as our first baseline. We extended this work to a multi-core algorithm, \textbf{\texttt{M-BBC}}, for exact counting of balanced butterflies in large-scale signed graphs, exploiting fine-grained parallelism to accelerate butterfly counting. (The full version is in ~\cite{kiran2026multi}).


\subsubsection{Distributed Baseline: \texttt{S-Monarch}}
To establish a distributed baseline, we adapt \texttt{Monarch}~\cite{tang2024monarch}, a distributed butterfly-counting algorithm originally designed for \emph{unsigned} bipartite graphs. We extend its counting procedure to account for edge signs, thereby identifying balanced butterflies in signed bipartite graphs; we refer to this adapted version as \textbf{\texttt{S-Monarch}}. While \texttt{S-Monarch} supports distributed balanced butterfly counting, its enumeration does not exploit the sign structure to prune unbalanced candidates early. As a result, it performs unnecessary work on candidate butterflies that cannot contribute to the balanced count, particularly in graphs with a high fraction of unbalanced butterflies.


To mitigate the above mentioned limitations, in the following section we present our proposed distributed algorithm \texttt{D-BBC}.

\subsection{Distributed Balanced Butterfly Counting Algorithm: D-BBC}

The proposed Distributed Balanced Butterfly Counting (\texttt{D-BBC}) framework consists of five sequential phases. Each phase addresses a key challenge in distributed graph processing, including scalable graph loading, global graph representation, workload balancing, communication-efficient local subgraph construction, and parallel balanced butterfly counting. The following subsections describe each phase in detail.

\definecolor{compcol}{RGB}{29,158,117}   
\definecolor{commcol}{RGB}{216,90,48}     
\definecolor{compcol}{RGB}{29,158,117}   
\definecolor{commcol}{RGB}{216,90,48}     
\definecolor{rankA}{RGB}{55,138,221}      
\definecolor{rankB}{RGB}{29,158,117}      
\definecolor{rankC}{RGB}{186,117,23}      
 
\begin{figure}
    \centering
\begin{tikzpicture}[
    font=\small, >={Stealth[length=1.6mm]},
    comp/.style  = {rounded corners=2pt, draw=compcol!80!black, fill=compcol!14, line width=0.6pt, align=center, minimum height=5mm, inner sep=2pt},
    comm/.style  = {rounded corners=2pt, draw=commcol!85!black, fill=commcol!16, line width=0.7pt, align=center, minimum height=5mm, inner sep=2pt},
    thr/.style   = {rounded corners=1.5pt, draw=compcol!80!black, fill=compcol!14, line width=0.5pt, align=center, minimum width=3mm, minimum height=4mm, inner sep=1pt},
    plabel/.style= {align=left, font=\footnotesize\itshape, text=black!65},
    note/.style  = {align=center, font=\footnotesize, text=black!62},
    fl/.style    = {->, line width=0.6pt, draw=black!55},
    bus/.style   = {line width=0.5pt, draw=black!55},
]
 
 
\node[comm, minimum width=62mm, minimum height=6.5mm] (file1) at (2.3,0.5)
  {\textbf{\footnotesize Input signed bipartite graph}\\[-2pt]{\footnotesize $|L|{+}|R|$ vertices,\; $|E|$ edges,\; $N$ bytes total}};
\draw[fl] (file1.south) -- ++(0,-6mm);
 
\node[note, anchor=west,font=\footnotesize] at (0.75,-0.22) {Byte-offset partitioning \\ ($N$ bytes\; $\rightarrow$\; ${\approx}\,N/P$ bytes per rank)};
 
\def\leftx{-0.8}
\def\rightx{5.4}
\def\sy{-0.85}
\def\hh{0.25}      

\fill[rankA!22] (\leftx,\sy-\hh) rectangle (1.0,\sy+\hh);
\fill[rankB!22] (1.0,\sy-\hh) rectangle (2.6,\sy+\hh);
\fill[rankC!24] (4.0,\sy-\hh) rectangle (\rightx,\sy+\hh);

\draw[rounded corners=2pt, draw=black!55, line width=0.7pt]
(\leftx,\sy-\hh) rectangle (\rightx,\sy+\hh);

\draw[black!45] (1.0,\sy-\hh)--(1.0,\sy+\hh);
\draw[black!45] (2.6,\sy-\hh)--(2.6,\sy+\hh);
\draw[black!45] (4.0,\sy-\hh)--(4.0,\sy+\hh);

\node[font=\small] at (0.3,\sy) {$[0,b_1)$};
\node[font=\small] at (1.7,\sy) {$[b_1,b_2)$};

\node[font=\large] at (3.0,\sy) {$\cdots$};

\node[font=\small] at (4.7,\sy) {$[b_{P-1},N)$};

\coordinate (b0) at (0.3,\sy-\hh);
\coordinate (b1) at (1.7,\sy-\hh);
\coordinate (b2) at (4.3,\sy-\hh);
 
\node[
    comp,minimum width=7mm, fill=rankA!22, draw=rankA!80!black
] (e0) at (0.5,-1.8) {Rank0};

\node[
    comp,
    minimum width=17mm,
    fill=rankB!22,
    draw=rankB!70!black
] (e1) at (2.2,-1.8) {$Rank_1$};

\node[note,font=\large] at (3.4,-1.8) {$\cdots$};

\node[
    comp,
    minimum width=7mm,
    fill=rankC!24,
    draw=rankC!80!black
] (e2) at (4.5,-1.8){$Rank_{P-1}$};

\draw[->, line width=0.8pt, draw=rankA!80!black] (b0) -- (b0 |- e0.north);
\draw[->, line width=0.8pt, draw=rankB!70!black] (b1) -- (b1 |- e1.north);
\draw[->, line width=0.8pt, draw=rankC!80!black] (b2) -- (b2 |- e2.north);
 
\node[note ,font=\footnotesize] at (4.9,-2.62) {Partition into \\ parsing ranges \\ (Multi-threaded compute) };
\node[
    thr,
    fill=gray!20,
    draw=gray!70
] (t0) at (1.6,-2.55) {\footnotesize $T_0$};

\node[
    thr,
    fill=gray!20,
    draw=gray!70
] (t1) at (2.0,-2.55) {\footnotesize $T_1$};

\node[
    thr,
    fill=gray!20,
    draw=gray!70
] (t2) at (2.95,-2.55) {\footnotesize $T_{k-1}$};

\node[font=\footnotesize,text=gray!70] at ($(t1)!0.5!(t2)$) {$\cdots$};
\draw[fl] (e1.south -| t0) -- (t0.north);
\draw[fl] (e1.south -| t1) -- (t1.north);
\draw[fl] (e1.south -| t2) -- (t2.north);
 
\node[comp, minimum width=8mm] (m1) at (2.3,-3.5) {\textbf{\footnotesize Parsed local edge list}\\[-2pt]{\footnotesize ($Rank_1$)}};
\coordinate (jy) at ($(t1.south)+(0,-1.6mm)$);
\draw[bus] (t0.south) |- (jy);
\draw[bus] (t2.south) |- (jy);
\draw[bus] (t1.south) -- (jy);
\draw[fl] (t1.south -| m1.north) ++(0,-1.6mm) -- (m1.north);

\end{tikzpicture}
    \caption{Phase 1}
    \label{fig:phase1}
\end{figure}

\subsubsection{Phase 1: Parallel Graph Loading}
\label{phase1}
As illustrated in Fig.~\ref{fig:phase1}, the proposed D-BBC framework
begins with a \emph{parallel graph loading} phase, in which the input signed bipartite graph is partitioned into multiple byte ranges, allowing each MPI process to independently read a distinct portion of the graph file in parallel. To ensure correctness, the boundaries of
each byte range are aligned to complete edge records, preventing graph edges from being split across adjacent MPI processes. Consequently, every edge is assigned to exactly one MPI process, ensuring that the entire graph is loaded without duplication.

After reading its assigned graph chunk, each MPI process further divides its local data into multiple parsing ranges proportional to the chunk size. These parsing ranges are processed concurrently using Intel Threading Building Blocks (TBB), where each thread independently
parses its assigned byte range and stores the extracted edges in a private thread-local buffer. Upon completion of all parsing tasks, the thread-local buffers are merged to construct the parsed local edge list for the corresponding MPI process. By combining inter-process
parallel file access with intra-process multithreaded parsing, the hybrid MPI-TBB strategy exploits both distributed-memory and shared-memory parallelism, alleviating the I/O bottleneck of single-process graph loading.

\begin{figure}[t]
\centering
\resizebox{1.0\columnwidth}{!}{%
\begin{tikzpicture}[
    font=\small, >={Stealth[length=1.2mm]},
    comp/.style  = {rounded corners=2pt, draw=compcol!80!black, fill=compcol!14, line width=0.6pt, align=center, minimum width=15mm, minimum height=3mm, inner sep=1.5pt},
    comm/.style  = {rounded corners=2pt, draw=commcol!85!black, fill=commcol!16, line width=0.7pt, align=center, minimum height=3mm, inner sep=1.5pt},
    plabel/.style= {align=left, font=\footnotesize\itshape, text=black!65},
    note/.style  = {align=center, font=\footnotesize, text=black!62},
    fl/.style    = {->, line width=0.35pt, draw=black!70!black},
    fll/.style   = {->, line width=0.35pt, draw=black!70!black},
]
\def\cA{0}\def\cB{2.3}\def\cC{4.6}

\node[note, font=\scriptsize] at (\cA,3.30) {\textbf{Rank$_0$}};
\node[note, font=\scriptsize] at (\cB,3.30) {\textbf{Rank$_1$}};
\node[note, font=\scriptsize] at (\cC,3.30) {\textbf{Rank$_{p-1}$}};
\node[note] at (3.45,2.85) {$\cdots$};

\node[comp, fill=rankA!22, draw=rankA!80!black] (a0) at (\cA,2.85) {\scriptsize Parsed local\\[-2pt]\scriptsize edge list};
\node[comp] (a1) at (\cB,2.85) {\scriptsize Parsed local\\[-2pt]\scriptsize edge list};
\node[comp, fill=rankC!24, draw=rankC!80!black] (a2) at (\cC,2.85) {\scriptsize Parsed local\\[-2pt]\scriptsize edge list};

\node[comp, fill=rankA!22, draw=rankA!80!black] (e0) at (\cA,1.9) {\scriptsize partial degree($v$),\\[-2pt]\scriptsize $\mathrm{owner}(v)$};
\node[comp] (e1) at (\cB,1.9) {\scriptsize partial degree($v$),\\[-2pt]\scriptsize $\mathrm{owner}(v)$};
\node[comp, fill=rankC!24, draw=rankC!80!black] (e2) at (\cC,1.9) {\scriptsize partial degree($v$),\\[-2pt]\scriptsize $\mathrm{owner}(v)$};
\draw[fl] (a0.south) -- (e0.north);
\draw[fl] (a1.south) -- (e1.north);
\draw[fl] (a2.south) -- (e2.north);

\node[comm, minimum width=4mm] (route) at (\cB,1.0)
  {\footnotesize \textsc{MPI\_Alltoallv}};
\node[note, font=\scriptsize] at (0,1.0) { Route vertex metadata \\[-2pt]
(owner($v$), partial degree($v$))};
  
\draw[fll] (e0.south) -- (route.north);
\draw[fll] (e1.south) -- (route.north);
\draw[fll] (e2.south) -- (route.north);

\node[comp, minimum width=4mm] (agg) at (\cB,0.10)
  {\scriptsize Aggregate partial degrees\\[-1pt]\scriptsize and computes global degree($v$)};
\node[note, font=\large] (vd1) at (-0.55,0.10) {$\vdots$};
\node[note, font=\large] (vd2) at (5,0.10) {$\vdots$};
\draw[fll] (route.south) -- (agg.north);

\node[comm, minimum width=4mm] (mail) at (\cB,-0.7)
  {\scriptsize \textsc{MPI\_Alltoallv}};
  
\draw[fll] (agg.south) -- (mail.north);
\draw[gray,fll,dotted] (route.south) -- ($(vd1.north)+(0,-2mm)$);
\draw[gray,fll,dotted] (route.south) -- ($(vd2.north)+(0,-2mm)$);

\node[note, font=\scriptsize] at (\cB+2,-0.7) { Return vertex degrees \\[-2pt]  back to requesters};


\node[comp, fill=rankA!22, draw=rankA!80!black] (g0) at (\cA,-1.6) {\scriptsize global degree($v$),\\[-1pt]\scriptsize $\mathrm{owner}(v)$};
\node[comp] (g1) at (\cB,-1.6) {\scriptsize global degree($v$),\\[-1pt]\scriptsize $\mathrm{owner}(v)$};
\node[comp, fill=rankC!24, draw=rankC!80!black] (g2) at (\cC,-1.6) {\scriptsize global degree($v$),\\[-1pt]\scriptsize $\mathrm{owner}(v)$};
\draw[fll] (mail.south) -- (g0.north);
\draw[fll] (mail.south) -- (g1.north);
\draw[fll] (mail.south) -- (g2.north);
\end{tikzpicture}%
}
\caption{Phase~2.}
\label{fig:phase2}
\end{figure}

\subsubsection{Phase 2: Global Degree Computation}
Since the input graph is distributed across MPI processes through byte-offset partitioning, the same vertex may be referenced by multiple processes, and vertex identifiers are not globally contiguous. Each process first deduplicates its local vertices and
computes their partial degrees. Every distinct vertex is then routed, through a single all-to-all exchange, to its owner process, defined as $\mathrm{owner}(v)=h(v)\bmod P$, where $h(\cdot)$ is a deterministic
multiplicative hash function and $P$ denotes the number of MPI processes. Each owner process consolidates duplicate occurrences of each vertex, aggregates the received partial degrees to compute its global degree, and assigns the vertex a unique compact identifier. A
lightweight collective communication is then performed to establish globally contiguous identifier ranges across all owner processes, ensuring that every vertex is assigned a globally unique identifier. Finally, each owner process returns the generated indexing information
only to the processes that originally referenced the corresponding vertex through a second all-to-all exchange. As illustrated in Fig.~\ref{fig:phase2}, the resulting distributed graph index provides globally consistent vertex identifiers and degrees, while each MPI
process stores indexing information only for the vertices appearing in its local graph partition.

\subsubsection{Phase 3: Workload Distribution}
As illustrated in Fig.~\ref{fig:phase3}, the objective of this phase is to balance the computational workload among MPI processes before distributed butterfly counting begins. Since the computational cost of processing a pivot vertex depends on the sizes of the neighborhoods visited during wedge enumeration, each process first estimates the workload of the pivot vertices appearing in its local graph partition.
The workload of a pivot vertex $u$ is estimated as
\begin{equation}
W(u)=\sum_{v\in\Gamma(u)} d(v),
\label{eq:workload}
\end{equation}
where $\Gamma(u)$ denotes the neighbors of $u$, and $d(v)$ is the degree of neighbor $v$. This metric approximates the computational effort required to process each pivot vertex during the local butterfly
counting phase. Each MPI process computes partial workload contributions for the pivot vertices referenced in its local partition. Since the same pivot vertex may appear in multiple partitions, the partial workload information is
routed to the corresponding owner process through an \textsc{MPI\_Alltoallv} communication. Each owner process aggregates the received partial workloads to obtain the global workload estimate for every pivot vertex under its ownership (as shown in example 1).

\begin{figure}[t]
\centering
\resizebox{1.0\columnwidth}{!}{%
\begin{tikzpicture}[
    font=\small, >={Stealth[length=1.2mm]},
    comp/.style  = {rounded corners=2pt, draw=compcol!80!black, fill=compcol!14, line width=0.6pt, align=center, minimum width=15mm, minimum height=3mm, inner sep=1.5pt},
    comm/.style  = {rounded corners=2pt, draw=commcol!85!black, fill=commcol!16, line width=0.7pt, align=center, minimum height=3mm, inner sep=1.5pt},
    plabel/.style= {align=left, font=\footnotesize\itshape, text=black!65},
    note/.style  = {align=center, font=\footnotesize, text=black!62},
    fl/.style    = {->, line width=0.35pt, draw=black!70!black},
    fll/.style   = {->, line width=0.35pt, draw=black!70!black},
]
\def\cA{0}\def\cB{2.3}\def\cC{4.6}

\node[note, font=\scriptsize] at (\cA,3.30) {\textbf{Rank$_0$}};
\node[note, font=\scriptsize] at (\cB,3.30) {\textbf{Rank$_1$}};
\node[note, font=\scriptsize] at (\cC,3.30) {\textbf{Rank$_{p-1}$}};
\node[note] at (3.45,2.85) {$\cdots$};

\node[comp, fill=rankA!22, draw=rankA!80!black] (a0) at (\cA,2.85) {\scriptsize Local graph\\[-2pt]\scriptsize global degrees};
\node[comp] (a1) at (\cB,2.85) {\scriptsize  Local graph\\[-2pt]\scriptsize global degrees};
\node[comp, fill=rankC!24, draw=rankC!80!black] (a2) at (\cC,2.85) {\scriptsize Local graph\\[-2pt]\scriptsize global degrees};

\node[comp, fill=rankA!22, draw=rankA!80!black] (e0) at (\cA,1.9) {\scriptsize estimate partial
,\\[-2pt]\scriptsize vertex workload};
\node[comp] (e1) at (\cB,1.9) {\scriptsize estimate partial
,\\[-2pt]\scriptsize vertex workload};
\node[comp, fill=rankC!24, draw=rankC!80!black] (e2) at (\cC,1.9) {\scriptsize estimate partial
,\\[-2pt]\scriptsize vertex workload};
\draw[fl] (a0.south) -- (e0.north);
\draw[fl] (a1.south) -- (e1.north);
\draw[fl] (a2.south) -- (e2.north);

\node[note,font=\scriptsize] at (0,1.0) { Route vertex metadata \\[-2pt] (partial workload ($v$) \\ [-2pt] \scriptsize to owners)};

\node[comm, minimum width=4mm] (agg) at (\cB,1.0)
  { \scriptsize \textsc{MPI\_Alltoallv}};

\draw[fll] (e0.south) -- (agg.north);
\draw[fll] (e1.south) -- (agg.north);
\draw[fll] (e2.south) -- (agg.north);

\node[comp, minimum width=4mm] (agg) at (\cB,0.15)
  {\scriptsize Aggregate Partial Workloads \\ [-2pt] \scriptsize
Compute Global Vertex Workloads};
\node[note, font=\scriptsize] (vd1) at (-0.55,0.30) {$\vdots$};
\node[note, font=\scriptsize] (vd2) at (5,0.30) {$\vdots$};
\draw[fll] (route.south) -- (agg.north);

\node[comm, minimum width=4mm] (mail) at (\cB,-0.75)
  {\scriptsize \textsc{MPI\_Gatherv}};
 \draw[fll] (agg.south) -- (mail.north); 
\draw[gray,fll,dotted] (route.south) -- ($(vd1.north)+(0,-2mm)$);
\draw[gray,fll,dotted] (route.south) -- ($(vd2.north)+(0,-2mm)$);

\draw[gray,fll,dotted] ($(vd1.south)+(-0,2mm)$) -- (mail.north);
\draw[gray,fll,dotted] ($(vd2.south)+(-0,2mm)$) -- (mail.north);

\node[comp, fill=rankA!22, draw=rankA!80!black] (g4) at (\cB,-1.8) {\scriptsize  $Rank_0$,\\[-2pt]\scriptsize Greedy bin-packing \\[-2pt]\scriptsize (Load balancing)};
\draw[fll] (mail.south) -- (g4.north);

\node[comp, fill=rankA!22, draw=rankA!80!black] (g0) at (\cA,-2.9) {\scriptsize owner's assigned \\[-2pt] \scriptsize ranks};
\node[comp] (g1) at (\cB,-2.9) {\scriptsize owner's assigned \\[-2pt] \scriptsize ranks};
\node[comp, fill=rankC!24, draw=rankC!80!black] (g2) at (\cC,-2.9) {\scriptsize owner's assigned \\[-2pt] \scriptsize ranks};
\draw[fll] (g4.south) -- (g0.north);
\draw[fll] (g4.south) -- (g1.north);
\draw[fll] (g4.south) -- (g2.north);

\node[comm, minimum width=4mm] (agg1) at (\cB,-3.7)
  { \scriptsize \textsc{MPI\_Alltoallv}};

  \draw[fll] (g0.south) -- (agg1.north);
\draw[fll] (g1.south) -- (agg1.north);
\draw[fll] (g2.south) -- (agg1.north);

\node[comp, fill=rankA!22, draw=rankA!80!black] (g00) at (\cA,-4.6) {\scriptsize pivot owner \\[-2pt] \scriptsize per vertex};
\node[comp] (g10) at (\cB,-4.6) {\scriptsize pivot owner \\[-2pt] \scriptsize per vertex};
\node[comp, fill=rankC!24, draw=rankC!80!black] (g20) at (\cC,-4.6) {\scriptsize pivot owner \\[-2pt] \scriptsize per vertex};

\node[note, font=\scriptsize] at (\cB+2.5,-3.7) { Return vertex loads \\[-2pt]  back to requesters};

\node[note, font=\scriptsize] at (\cB+2,-2.2) { \textsc{MPI\_Scatterv}};

\draw[fll] (agg1.south) -- (g00.north);
\draw[fll] (agg1.south) -- (g20.north);
\draw[fll] (agg1.south) -- (g10.north);

\end{tikzpicture}%
}
\caption{Phase~3.}
\label{fig:phase3}
\end{figure}

The aggregated workload is then collected at Rank$_0$ using \textsc{MPI\_Gatherv}. Since workload distribution is performed only once prior to the counting phase, the scheduling overhead is negligible
compared with the overall execution time. Therefore, a centralized scheduler is employed to simplify workload management and to construct a globally balanced assignment without requiring iterative coordination
among MPI processes. Based on the global workload information, Rank$_0$ performs a greedy bin-packing strategy by sorting pivot vertices in descending order of workload and repeatedly assigning each pivot vertex
to the MPI process with the minimum accumulated workload (as shown in example 2). The resulting pivot assignments are distributed to the owner processes using \textsc{MPI\_Scatterv}. Since multiple processes may reference the same pivot vertex, each owner process then forwards its assigned MPI rank to the corresponding requesters via a second \textsc{MPI\_Alltoallv} communication.

At the end of this phase, every MPI process knows the assigned process of each pivot vertex required for subsequent computation. This globally consistent, workload-aware pivot distribution enables balanced
computation while minimizing idle time during the distributed counting phase.

\textbf{Example 1 (Workload estimation):} 
\label{ex:we}
Consider a pivot vertex $u$ with neighbors $\Gamma(u)=\{v_1,v_2,v_3\}$ having degrees $4$, $6$, and $5$, respectively. Using Eq.~(\ref{eq:workload}), its workload is
\[
W(u)=d(v_1)+d(v_2)+d(v_3)=4+6+5=15.
\]
After aggregating workloads from all MPI processes, the greedy scheduler assigns this pivot vertex to the MPI process with the minimum accumulated workload, thereby improving load balance across the distributed system.

\textbf{Example 2 (Greedy Bin-Packing):}
\label{ex:gbp}

Consider six pivot vertices with workloads
$\{15,12,9,7,6,4\}$ to be assigned to $P=3$ MPI processes.
Using a naive vertex-based assignment, where consecutive pivot vertices are assigned to each process without considering their computational costs, the resulting assignments are $R_0=\{u_1,u_2\}$, $R_1=\{u_3,u_4\}$, and $R_2=\{u_5,u_6\}$. The corresponding workloads are $27$, $16$, and $10$, respectively, yielding a straggler ratio of \textbf{$27/10 = 2.7$}. This indicates a significant workload imbalance: the process assigned the heaviest workload becomes the execution bottleneck, while the remaining processes remain underutilized.

In contrast, the proposed workload-aware greedy bin-packing strategy first sorts the pivot vertices in descending order of their estimated workloads and iteratively assigns each pivot to the process currently with the least load. This results in the assignments $R_0=\{u_1,u_6\}$, $R_1=\{u_2,u_5\}$, and $R_2=\{u_3,u_4\}$, with final workloads of $19$, $18$, and $16$, respectively. Consequently, the straggler ratio is reduced to $19/16 \approx 1.19$, demonstrating that the proposed workload-aware assignment substantially improves load balance and minimizes the likelihood of idle processes waiting for the slowest MPI process to complete.

\subsubsection{Phase 4: Distributed Local Subgraph Construction}

\begin{figure}[t]
\centering
\resizebox{1.0\columnwidth}{!}{%
\begin{tikzpicture}[
    font=\small, >={Stealth[length=1.2mm]},
    comp/.style  = {rounded corners=2pt, draw=compcol!80!black, fill=compcol!14, line width=0.6pt, align=center, minimum width=15mm, minimum height=3mm, inner sep=1.5pt},
    comm/.style  = {rounded corners=2pt, draw=commcol!85!black, fill=commcol!16, line width=0.7pt, align=center, minimum height=3mm, inner sep=1.5pt},
    plabel/.style= {align=left, font=\footnotesize\itshape, text=black!65},
    note/.style  = {align=center, font=\footnotesize, text=black!62},
    fl/.style    = {->, line width=0.35pt, draw=black!70!black},
    fll/.style   = {->, line width=0.35pt, draw=black!70!black},
]

\def\cA{0}\def\cB{2.3}\def\cC{4.6}

\node[note, font=\footnotesize] at (\cA,3.40) {\textbf{Rank$_0$}};
\node[note, font=\footnotesize] at (\cB,3.40) {\textbf{Rank$_1$}};
\node[note, font=\footnotesize] at (\cC,3.40) {\textbf{Rank$_{p-1}$}};
\node[note] at (3.45,2.85) {$\cdots$};

\node[comp, fill=rankA!22, draw=rankA!80!black] (a0) at (\cA,2.85) {\footnotesize pivot owner \\[-2pt] \footnotesize per vertex};
\node[comp] (a1) at (\cB,2.85) {\footnotesize pivot owner \\[-2pt] \footnotesize per vertex};
\node[comp, fill=rankC!24, draw=rankC!80!black] (a2) at (\cC,2.85) {\footnotesize pivot owner \\[-2pt] \footnotesize per vertex};

\node[note] at (-0.3,1.8) { Distribute edges according \\ [-2pt] \footnotesize to pivot ownership};

\node[comm, minimum width=4mm] (agg) at (\cB,1.8)
  { \footnotesize \textsc{MPI\_Alltoallv}};

\draw[fll] (a0.south) -- (agg.north);
\draw[fll] (a1.south) -- (agg.north);
\draw[fll] (a2.south) -- (agg.north);

\node[comp, fill=rankA!22, draw=rankA!80!black] (g0) at (\cA-0.5,0.45) {\footnotesize receive all edges  \\[-2pt] \footnotesize of the pivots \\[-2pt] \footnotesize assigned to this process};
\node[comp] (g1) at (\cB,0.45) {\footnotesize receive all edges  \\[-2pt] \footnotesize of the pivots \\[-2pt] \footnotesize assigned to this process};
\node[comp, fill=rankC!24, draw=rankC!80!black] (g2) at (\cC+0.5,0.45) {\footnotesize receive all edges  \\[-2pt] \footnotesize of the pivots \\[-2pt] \footnotesize assigned to this process};
\draw[fll] (agg.south) -- (g0.north);
\draw[fll] (agg.south) -- (g1.north);
\draw[fll] (agg.south) -- (g2.north);

\node[comm, minimum width=4mm] (agg1) at (\cB,-0.8)
  { \footnotesize \textsc{MPI\_Alltoallv}};

  \draw[fll] (g0.south) -- (agg1.north);
\draw[fll] (g1.south) -- (agg1.north);
\draw[fll] (g2.south) -- (agg1.north);

\node[comp, fill=rankA!22, draw=rankA!80!black] (g00) at (\cA,-1.8) {\footnotesize construct \\[-2pt] \footnotesize local subgraph};
\node[comp] (g10) at (\cB,-1.8) {\footnotesize construct \\[-2pt] \footnotesize local subgraph};
\node[comp, fill=rankC!24, draw=rankC!80!black] (g20) at (\cC,-1.8) {\footnotesize construct \\[-2pt] \footnotesize local subgraph};

\node[note] at (\cB+3,-0.8) { Exchange adjacent vertices \\ [-2pt] \footnotesize required for 2-hop \\ [-2pt] \footnotesize subgraph construction};

\draw[fll] (agg1.south) -- (g00.north);
\draw[fll] (agg1.south) -- (g20.north);
\draw[fll] (agg1.south) -- (g10.north);

\end{tikzpicture}%
}
\caption{Phase~4.}
\label{fig:phase4}
\end{figure}

As illustrated in Fig.~\ref{fig:phase4}, this phase constructs, on every MPI process, the local subgraph required to count the butterflies anchored at its assigned pivot vertices. The workload-aware pivot
assignment produced in Phase~3 specifies, for each vertex, the MPI process responsible for processing it. However, the edges incident to a pivot vertex may be scattered across several processes due to the byte-offset partitioning in Phase~1. Before counting can begin, all edges and adjacency information relevant to a pivot must be gathered on the process to which that pivot is assigned.

Guided by the pivot assignment, each MPI process routes every local edge to the process that owns its pivot endpoints through an \textsc{MPI\_Alltoallv} communication. After this exchange, each process holds the complete set of edges incident to the pivot vertices assigned to it, i.e., the one-hop neighborhood of each of its pivots. Because butterfly enumeration requires closing wedges of the form $u\!\rightarrow\!w\!\rightarrow\!v$, the two-hop neighborhood of each pivot is also required. A second \textsc{MPI\_Alltoallv} communication, therefore, exchanges the adjacency lists of these intermediate vertices so that every process obtains the two-hop adjacency needed to enumerate all wedges centered at its assigned pivots.

Once all required adjacency information has been received, each MPI process locally assembles its subgraph in a compact adjacency (CSR) representation, indexing its assigned pivots together with their one-hop
and two-hop neighbors. This distributed construction ensures that every process holds exactly the portion of the graph needed for its assigned pivots, and no more, so that the subsequent butterfly counting phase can
proceed entirely on local data without any further communication. By localizing all adjacency information ahead of counting, the proposed design confines communication to these two well-defined exchange steps
and eliminates fine-grained remote access during the computation-intensive counting phase.

\subsubsection{Phase 5: Parallel Local Balanced Butterfly Counting}
\label{phase5}
As illustrated in Fig.~\ref{fig:phase5}, each MPI process counts the balanced butterflies anchored at its assigned pivot vertices using the local subgraph constructed in Phase~4. Since all required one-hop and two-hop neighborhood information has already been localized, this phase is performed entirely on local data without any inter-process communication.

\definecolor{compcol}{RGB}{29,158,117}
\begin{figure}[t]
\centering
\resizebox{1.0\columnwidth}{!}{%
\begin{tikzpicture}[
    font=\small, >={Stealth[length=1.2mm]},
    comp/.style  = {rounded corners=2pt, draw=compcol!80!black, fill=compcol!14, line width=0.6pt, align=center, minimum width=15mm, minimum height=3mm, inner sep=1.5pt},
    comm/.style  = {rounded corners=2pt, draw=commcol!85!black, fill=commcol!16, line width=0.7pt, align=center, minimum height=3mm, inner sep=1.5pt},
    thr/.style   = {rounded corners=1.2pt, draw=black!55, fill=black!12, line width=0.4pt, align=center, minimum width=4mm, minimum height=4mm, inner sep=0.5pt, font=\footnotesize},
    tbbblock/.style = {rounded corners=2pt, draw=compcol!80!black, fill=compcol!14,  line width=0.6pt, align=center},
    tiny2/.style = {rounded corners=1pt, draw=black!45, fill=black!8, line width=0.3pt, align=center, font=\footnotesize, inner sep=0.8pt},
    note/.style  = {align=center, font=\footnotesize, text=black!62},
    fl/.style    = {->, line width=0.35pt, draw=black!70},
    fll/.style   = {->, line width=0.35pt, draw=black!70},
    micro/.style = {->, line width=0.3pt, draw=black!60},
    wiggle/.style = {->, line width=0.35pt, draw=black!60, decorate,
                     decoration={snake, amplitude=0.35mm, segment length=1.3mm, post length=1mm}},
]
\def\cA{0}\def\cB{2.3}\def\cC{4.6}
\node[note, font=\scriptsize] at (\cA,3.30) {\textbf{Rank$_0$}};
\node[note, font=\scriptsize] at (\cB,3.30) {\textbf{Rank$_1$}};
\node[note, font=\scriptsize] at (\cC,3.30) {\textbf{Rank$_{p-1}$}};
\node[note, font = \large ] at (3.45,2.75) {$\cdots$};
\node[comp, fill=rankA!22, draw=rankA!80!black] (a01) at (\cA,2.75) {\footnotesize local subgraph};
\node[comp] (a11) at (\cB,2.75) {\footnotesize local subgraph};
\node[comp, fill=rankC!24, draw=rankC!80!black] (a21) at (\cC,2.75) {\footnotesize local subgraph};
\draw[fl] (a0) -- (a01);
\draw[fl] (a1) -- (a11);
\draw[fl] (a2) -- (a21);
 
\node[tbbblock, minimum width=22mm, minimum height=24mm] (b1) at (\cB,1.15) {};
\node[note, anchor=north] at ([yshift=-0.4mm]b1.north) {};

\node[
    thr,
    fill=gray!20,
    draw=gray!70
] (t0) at (1.45,2) {\footnotesize $T_0$};

\node[
    thr,
    fill=gray!20,
    draw=gray!70
] (t1) at (1.95,2) {\footnotesize $T_1$};
\node[font=\large] at (2.5,2) {$\cdots$};
\node[
    thr,
    fill=gray!20,
    draw=gray!70
] (t2) at (2.95,2) {\footnotesize $T_{k-1}$};

\draw[wiggle] (t0.south) -- ++(0,-0.4);

\draw[wiggle]
    (t1.south) -- coordinate (w1) ++(0,-0.6);

\draw[wiggle] (t2.south) -- ++(0,-0.4);

\node[note, font=\large] (vd1) at (-0.55,0.60) {$\vdots$};
\node[note, font=\large] (vd2) at (5,0.60) {$\vdots$};
\node[note] at (\cB+2,1.1) { \footnotesize by following \\[-1pt] \footnotesize Definition~\ref{def-priority}};

\node[tiny2] (wd) at ([yshift=-4mm]w1) {$(u,w,v)$};
\node[tiny2, fill=compcol!12, draw=compcol!70!black] (sym)  at ($(wd)+(-0.35,-0.42)$) {sym};
\node[tiny2, fill=commcol!14, draw=commcol!75!black] (asym) at ($(wd)+(0.35,-0.42)$)  {asym};
\draw[micro] (wd.south) -- (sym.north);
\draw[micro] (wd.south) -- (asym.north);
\node[tiny2] (ch) at ($(wd)+(0.2,-0.9)$) {$\binom{\cdot}{2}$ choose\,2};
\draw[micro] (sym.south)  -- (ch.north);
\draw[micro] (asym.south) -- (ch.north);
\draw[fl] (a11) -- (b1.north);
 
\node[comp, fill=rankA!22, draw=rankA!80!black, minimum width=13mm] (c0) at (\cA,-0.6) {\footnotesize local count};
\node[comp, minimum width=13mm] (c1) at (\cB,-0.6) {\footnotesize local count};
\node[comp, fill=rankC!24, draw=rankC!80!black, minimum width=13mm] (c2) at (\cC,-0.6) {\footnotesize local count};
 
\draw[fl] ($(ch.south)+(0.15,0)$) -- (c1.north);
 
\node[comp, fill=compcol!20, minimum width=26mm] (tot) at (\cB,-1.30) {\footnotesize total balanced butterflies};
\draw[fll] (c0.south) -- (tot.north);
\draw[fll] (c1.south) -- (tot.north);
\draw[fll] (c2.south) -- (tot.north);

\end{tikzpicture}%
}
\caption{Phase~5.}
\label{fig:phase5}
\end{figure}

For each assigned pivot vertex, the process traverses its two-hop neighborhood to enumerate wedges. To ensure that every balanced butterfly is counted exactly once, wedge traversal follows the vertex priority rule defined in Definition~\ref{def-priority}, thereby eliminating duplicate enumeration across MPI processes. Each wedge is classified according to the sum of its two edge signs. For every endpoint vertex, two bucket counters are maintained: $B_1$ stores symmetric wedges (sign-sum $0$ or $2$), whereas $B_2$ stores asymmetric wedges (sign-sum $1$). Since a balanced butterfly is
formed by pairing two wedges of the same type, the accumulated bucket counts are used to compute the number of balanced butterflies.

To exploit shared-memory parallelism, the assigned pivot vertices are distributed among Intel TBB threads. Each thread processes a disjoint subset of pivots and maintains private bucket counters, thereby eliminating the need for synchronization during enumeration. After all threads complete, their local counts are combined to obtain the process-local butterfly count. Finally, a single \textsc{MPI\_Reduce} operation aggregates the local counts from all MPI processes to produce the total
number of balanced butterflies in the input graph.

\begin{algorithm}[t]
\small
\DontPrintSemicolon
\SetKwInOut{Input}{Input}
\SetKwInOut{Output}{Output}
\SetKwFor{ForPar}{parallel for}{do}{end}

\Input{Signed bipartite graph $G=(U,V,E)$, $P$ MPI ranks, and $T$ threads per rank}

\Output{$\beta$: Total number of balanced butterflies $\beta$}

\textbf{Phase 1: Parallel graph loading}\;
Read the assigned graph partition in parallel\;

\BlankLine

Construct the global vertex index and compute the global vertex degrees
$d(\cdot)$ \tcp*{Phase 2}


\ForEach{vertex $x\in U\cup V$}{
    $w_{load}(x)\gets\displaystyle\sum_{y\in\Gamma(x)}d(y)$  \tcp*[r]{Phase 3}
}
Sort vertices by decreasing $w(\cdot)$\;
\ForEach{vertex $x$ in sorted order}{
    Assign $x$ to the least-loaded rank\;
}

Exchange aggregated higher-priority neighborhoods and construct the local subgraph \tcp*{Phase 4} 


$\beta_r\gets0$ \tcp*{Phase 5}

\ForPar{pivot $u$ assigned to the current rank}{
    Initialize thread-local buckets $B_1$ and $B_2$\;

    \ForEach{$v\in\Gamma(u)$ such that $p(v)<p(u)$}{

        \ForEach{$w\in\Gamma(v)$ such that $p(w)<p(u)$}{

            \eIf{$s(u,v)=s(v,w)$}{
                $B_1[w]\gets B_1[w]+1$\;
            }{
                $B_2[w]\gets B_2[w]+1$\;
            }

        }
    }

    \ForEach{touched endpoint $w$}{

        $\beta_r\gets\beta_r+
        \dbinom{B_1[w]}{2}+
        \dbinom{B_2[w]}{2}$\;

    }

    Reset the bucket entries corresponding to all touched endpoints\;

}

\BlankLine


$\beta\gets\textsc{Reduce}(\beta_r,+)$ \tcp*{Phase 6}

\Return $\beta$\;

\caption{D-BBC: Distributed balanced butterfly counting algorithm.}
\label{alg:dbbc}
\end{algorithm}

  \definecolor{rankA}{RGB}{55,138,221}
  \definecolor{rankB}{RGB}{29,158,117}
  \definecolor{rankC}{RGB}{186,117,23}
  \definecolor{rankD}{RGB}{178,80,82}
  \definecolor{idxcol}{RGB}{120,160,60}
  \definecolor{balcol}{RGB}{110,90,190}
  \definecolor{redcol}{RGB}{178,70,52}
\begin{figure*}[t]
\centering
\resizebox{1.45\columnwidth}{!}{%
\begin{tikzpicture}[
    font=\small, >={Stealth[length=1.6mm]},
    rk/.style={rounded corners=3pt, line width=0.7pt, align=center, inner sep=2.5pt},
    band/.style={rounded corners=3pt, line width=0.7pt, align=center, inner sep=3pt},
    plabel/.style={align=left, font=\footnotesize, text=black!60},
    psub/.style={align=left, font=\scriptsize, text=black!45},
    fl/.style={->, line width=0.6pt, draw=black!45},
]



\def\cA{0}\def\cB{3.4}\def\cC{6.8}\def\cD{10.2}
\def\midx{5.1}   

\tikzset{
  blackdot/.style={
    circle, draw=black, fill=blue!40,
    minimum size=6pt, inner sep=0pt
  },
  whitedot/.style={
    circle, draw=black, fill=red!40,
    minimum size=6pt, inner sep=0pt
  },
  edge/.style={
    line width=0.8pt,
    shorten >=1pt,
    shorten <=1pt
  }
}

\node[blackdot,label=left:$7$] (b7) at (4.55,13.15) {};
\node[blackdot,label=left:$2$] (b2) at (4.55,12) {};

\node[whitedot,label=right:$7$] (b5) at (5.65,13.15) {};
\node[whitedot,label=right:$8$] (b8) at (5.65,12) {};

\draw[red, edge] (b7)--(b5);
\draw[red, edge] (b7)--(b8);
\draw[red, edge] (b2)--(b5);
\draw[red, edge] (b2)--(b8);

\node[font=\itshape] at (8.1,12.45) {Sample Input Graph};
\node[rounded corners=6pt, draw=black!45, fill=black!5, line width=0.7pt,
      inner sep=5pt, font=\itshape, text=black] at (\midx,11.1)
      {Graph: edges $(2,5,0)$,\ \ $(2,8,0)$,\ \ $(7,5,0)$,\ \ $(7,8,0)$};
 
\node[font=\large] at (\cA,10.2) {Rank0};
\node[font=\large] at (\cB,10.2) {Rank1};
\node[font=\large] at (\cC,10.2) {Rank2};
\node[font=\large] at (\cD,10.2) {Rank3};
 
\node[plabel] at (-3.2,9.5) {Phase 1};
\node[rk, draw=rankA!40!black, fill=rankA!44, text=black, minimum width=24mm,font=\large] (e0) at (\cA,9.4) {edge (2,5)};
\node[rk, draw=rankB!70!black, fill=rankB!44, text=black, minimum width=24mm,font=\large] (e1) at (\cB,9.4) {edge (2,8)};
\node[rk, draw=rankC!80!black, fill=rankC!46, text=black, minimum width=24mm,font=\large] (e2) at (\cC,9.4) {edge (7,5)};
\node[rk, draw=rankD!80!black, fill=rankD!44, text=black, minimum width=24mm,font=\large] (e3) at (\cD,9.4) {edge (7,8)};
 
\node[plabel] at (-3.2,8.3) {Phase 2};
\node[band, draw=idxcol!75!black, fill=idxcol!44, text=black, minimum width=132mm,font=\large] (idx) at (\midx,8.2)
  {global degrees:\ \ $\deg(2)=\deg(7)=\deg(5)=\deg(8)=2$};
\foreach \e in {e0,e1,e2,e3} \draw[fl] (\e.south) -- (\e.south |- idx.north);
 
\node[plabel] at (-3.2,6.95) {Phase 3};
\node[band, draw=balcol!75!black, fill=balcol!42, text=black, minimum width=132mm, minimum height=13mm] (bal) at (\midx,6.8)
  {{\large greedy bin-pack $\rightarrow$ pivot assignment}\\[1pt]
   $2 \rightarrow$ Rank3\ \ $\cdot$\ \ $5 \rightarrow$ Rank0\ \ $\cdot$\ \ $7 \rightarrow$ Rank2\ \ $\cdot$\ \ $8 \rightarrow$ Rank1};
\draw[fl] (idx.south) -- (bal.north);
 
\node[plabel] at (-3.2,5.0) {Phase 4};
\node[rk, draw=rankA!80!black, fill=rankA!14, text=black, minimum width=27mm, minimum height=14mm] (s0) at (\cA,5.0)
  {\textbf{Rank0: piv 5}\\[1pt]$5\!\to\!\{2,7\}$\\[-1pt]{\scriptsize $2\!\to\!\{5,8\}$\ $7\!\to\!\{5,8\}$}};
\node[rk, draw=rankB!70!black, fill=rankB!14, text=black, minimum width=27mm, minimum height=14mm] (s1) at (\cB,5.0)
  {\textbf{Rank1: piv 8}\\[1pt]$8\!\to\!\{2,7\}$\\[-1pt]{\scriptsize $2\!\to\!\{5,8\}$\ $7\!\to\!\{5,8\}$}};
\node[rk, draw=rankC!80!black, fill=rankC!16, text=black, minimum width=27mm, minimum height=14mm] (s2) at (\cC,5.0)
  {\textbf{Rank2: piv 7}\\[1pt]$7\!\to\!\{5,8\}$\\[-1pt]{\scriptsize $5\!\to\!\{2,7\}$\ $8\!\to\!\{2,7\}$}};
\node[rk, draw=rankD!80!black, fill=rankD!14, text=black, minimum width=27mm, minimum height=14mm] (s3) at (\cD,5.0)
  {\textbf{Rank3: piv 2}\\[1pt]$2\!\to\!\{5,8\}$\\[-1pt]{\scriptsize $5\!\to\!\{2,7\}$\ $8\!\to\!\{2,7\}$}};
\draw[fl] (bal.south -| s0) -- (s0.north);
\draw[fl] (bal.south -| s1) -- (s1.north);
\draw[fl] (bal.south -| s2) -- (s2.north);
\draw[fl] (bal.south -| s3) -- (s3.north);
 
\node[plabel] at (-3.2,3.6) {Phase 5};
\node[rk, draw=rankA!80!black, fill=rankA!14, text=black, minimum width=27mm, minimum height=22mm] (p0) at (\cA,2.6)
  {{\large $u=5$}\\[2pt]mid 2 $(2\!<\!5$\ding{51}$)$\\ v 7 $(7\!<\!5$\ding{55}$)$\\ mid 7 $(7\!<\!5$\ding{55}$)$\\[2pt]no pair\\[3pt]{\large count 0}};
\node[rk, draw=rankB!70!black, fill=rankB!14, text=black, minimum width=27mm, minimum height=22mm] (p1) at (\cB,2.6)
  {{\large $u=8$}\\[2pt]mid 2 $(2\!<\!8$\ding{51}$)$\\ $\to 8\!\cdot\!2\!\cdot\!5$\\ mid 7 $(7\!<\!8$\ding{51}$)$\\ $\to 8\!\cdot\!7\!\cdot\!5$\\ $c_0[5]=2$\\[2pt]{\large count 1 $\star$}};
\node[rk, draw=rankC!80!black, fill=rankC!16, text=black, minimum width=27mm, minimum height=22mm] (p2) at (\cC,2.6)
  {{\large $u=7$}\\[2pt]mid 5 $(5\!<\!7$\ding{51}$)$\\ v 2 $(2\!<\!7$\ding{51}$)$\\ mid 8 $(8\!<\!7$\ding{55}$)$\\[2pt]$c_0[2]=1$\\[3pt]{\large count 0}};
\node[rk, draw=rankD!80!black, fill=rankD!14, text=black, minimum width=27mm, minimum height=22mm] (p3) at (\cD,2.6)
  {{\large $u=2$}\\[2pt]mid 5 $(5\!<\!2$\ding{55}$)$\\ mid 8 $(8\!<\!2$\ding{55}$)$\\[2pt]all pruned\\[6pt]{\large count 0}};
\foreach \a/\b in {s0/p0,s1/p1,s2/p2,s3/p3} \draw[fl] (\a.south) -- (\b.north);
 
\node[band, draw=redcol!80!black, fill=redcol!42, text=black, minimum width=60mm, minimum height=11mm,font=\large] (red) at (\midx,0.15)
  {{\large Reduce($+$):\ \ $0 + 1 + 0 + 0$}\\[1pt]{\large one global reduction}};
\foreach \p in {p0,p1,p2,p3} \draw[fl] (\p.south) -- (red.north);
 
\node[band, draw=balcol!75!black, fill=balcol!12, text=black, minimum width=48mm, minimum height=9mm,font=\large] (tot) at (\midx,-1.4)
  {{\normalsize 1 balanced butterfly}};
\draw[fl] (red.south) -- (tot.north);
 
\node[font=\itshape] at (\midx,-2.25) {the butterfly 2--5--7--8 is counted exactly once, at its highest-id pivot $u=8$};
\end{tikzpicture}%
}
\caption{End-to-end trace of the five-phase pipeline on a $K(2,2)$ example}
\label{fig:pipeline_example}
\end{figure*}
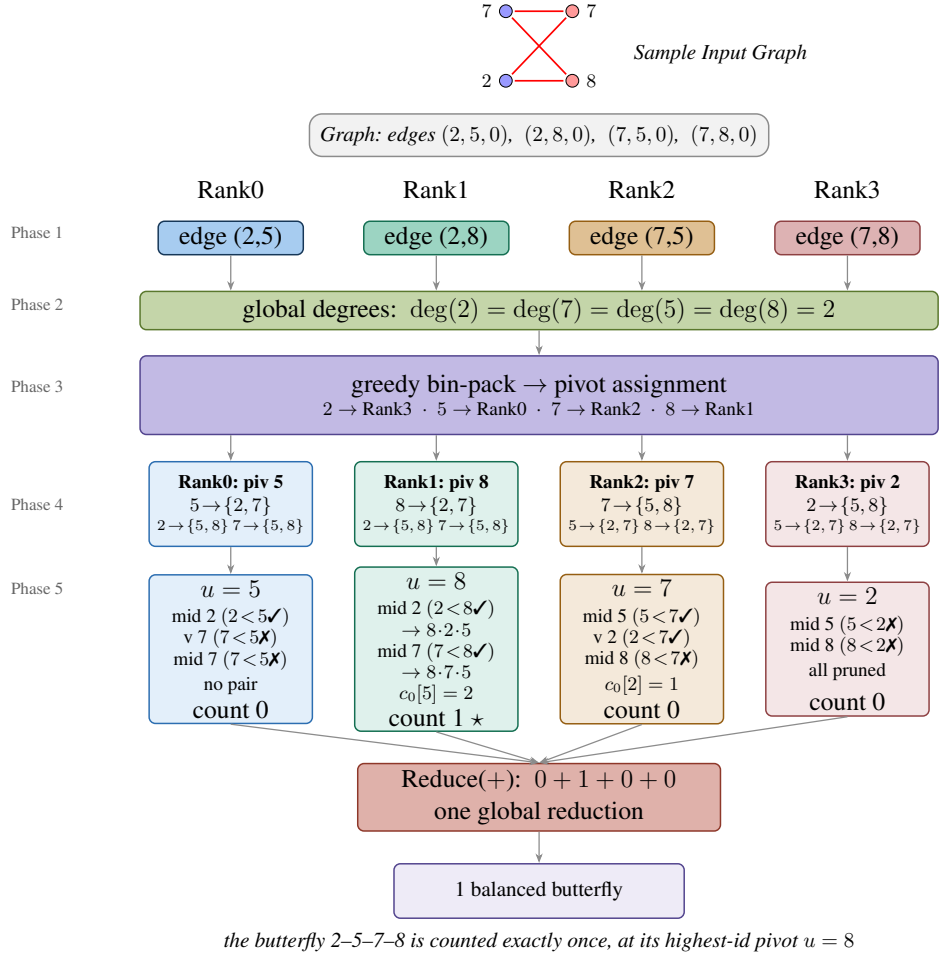

Algorithm~\ref{alg:dbbc} summarizes the complete workflow of the proposed \texttt{D-BBC} framework, consolidating the five phases described above into a
single procedure, with phase boundaries marked by side comments
corresponding to the subsections above.

\begin{theorem}
Algorithm~\ref{alg:dbbc} correctly computes the total number of balanced
butterflies in the signed bipartite graph $G$.
\label{thm:correctness}
\end{theorem}

\begin{proof}
We prove that Algorithm~\ref{alg:dbbc} counts every balanced butterfly exactly once.

In Phase~1, the input edge list is partitioned into disjoint byte ranges whose boundaries are aligned with complete edge records; hence every edge of $E$ is read by exactly one MPI. In Phase~2, the partial degrees produced by different partitions are aggregated at the unique owner of each vertex, yielding the exact global degree. In Phase~3, the workload of each pivot is computed from its complete neighborhood, and the greedy assignment maps every pivot to exactly one MPI rank. The resulting assignment associates every pivot with exactly one MPI rank. Phase~4, then exchanges the edges and adjacency information required by each assigned pivot, so each rank obtains the complete one-hop and two-hop neighborhood needed for its assigned pivots.

It remains to show that Phase~5 counts exactly the balanced butterflies. For an assigned pivot $u$, the priority condition $p(v)<p(u)$ and $p(w)<p(u)$ ensure that each butterfly has a unique highest-priority pivot and is therefore processed by exactly one rank. For every common endpoint $v$, the algorithm classifies each wedge according to whether its two edge signs are equal or different, storing them in $B_1[w]$ and $B_2[w]$, respectively. By Lemma~\ref{lemma:lem-1}, a balanced butterfly is formed if and only if its two wedges belong to the same class: either two symmetric wedges or two asymmetric wedges. Thus, $\binom{B_1[w]}{2}+\binom{B_2[w]}{2}$ counts exactly the balanced butterflies associated with endpoint $w$.
\end{proof}

\subsection{Complexity Analysis}
Let $n=|U|+|V|$, $m=|E|$, and let $P$ and $T$ denote the numbers of MPI ranks
and threads per rank, respectively.

\textbf{Computation.}
Phase~1 reads $O(N/P)$ bytes and parses $O(m/(PT))$ edges per thread.
Phase~2 deduplicates local vertices and aggregates degrees in expected
$O(m/P)$ time per rank via hashing. In Phase~3, partial workloads are
obtained in $O(m/P)$ time, and the centralized scheduler sorts the $n$
workloads in $O(n\log n)$ and performs the heap-based greedy assignment in
$O(n\log P)$. Phase~5 dominates: under the degree-based priority of
Definition~3, each edge $(u,v)$ is expanded only from its higher-priority
endpoint at a cost equal to the degree of its lower-priority endpoint, so the
total enumeration work is
$O\bigl(\sum_{(u,v)\in E}\min(d(u),d(v))\bigr)$.

\textbf{Communication.}
All exchanges are sparse, since payload flows only between ranks that share
graph data. In Phases~2 and~3, each rank routes $O(m/P)$ vertex and workload
records to their owner ranks, for an aggregate volume of $O(m)$, plus $O(n)$
words for the scheduler's gather and scatter operations. In Phase~4, each
edge travels only to the ranks owning its (at most two) pivot endpoints, and
adjacency lists are sent only to the ranks whose assigned pivots require
them, yielding $O(m)$ words plus at most the enumeration volume
$O\bigl(\sum_{(u,v)\in E}\min(d(u),d(v))\bigr)$. Phase~5 requires a single
reduction of $O(\log P)$ cost.

 \subsection{Ilustrative Example}
To illustrate the complete pipeline, we trace the execution of D-BBC on a small complete bipartite graph $(2,2)$ using $P=4$ MPI processes, as shown in Fig.~\ref{fig:pipeline_example}. The graph consists of four signed edges, $(2,5,0)$, $(2,8,0)$, $(7,5,0)$, and
$(7,8,0)$, where $\{2,7\}$ are left vertices and $\{5,8\}$ are right vertices, forming exactly one balanced butterfly.

\textit{\textbf{Phase 1.}} The four edges are distributed one per process through byte-offset partitioning: Rank$_0$ reads $(2,5,0)$, Rank$_1$ reads $(2,8,0)$, Rank$_2$ reads $(7,5,0)$, and Rank$_3$ reads $(7,8,0)$.

\textit{\textbf{Phase 2}.} Each vertex participates in two edges, so the global degree computation yields
$d(2)=d(7)=d(5)=d(8)=2$.

\textit{\textbf{Phase 3.}} Using the workload estimate of Eq.~(\ref{eq:workload}), all pivots have equal workload, and the greedy scheduler produces a balanced assignment: vertex $5\!\rightarrow\!$ Rank$_0$, $8 \rightarrow\!$ Rank$_1$, $7\!\rightarrow\!$ Rank$_2$, and
$2\!\rightarrow\!$ Rank$_3$.

\textit{\textbf{Phase 4}.} Each process gathers the one-hop and two-hop adjacency of its assigned pivot. For instance, Rank$_1$, which is assigned pivot $8$, receives $8\!\rightarrow\!\{2,7\}$ together with the adjacency $2\!\rightarrow\!\{5,8\}$ and $7\!\rightarrow\!\{5,8\}$ needed to close the wedges centered at $8$.

\textit{\textbf{Phase 5}.} Each process enumerates the wedges anchored at its pivot under the priority rule of Definition~\ref{def-priority}, which (with equal degrees) retains only intermediate and endpoint vertices whose identifier is smaller than the pivot. Consequently, the single butterfly is counted at exactly one canonical pivot, its highest-identifier vertex $u=8$:
\begin{itemize}
\item At Rank$_1$ ($u=8$), both wedges $8\!\rightarrow\!2\!\rightarrow\!5$ and $8\!\rightarrow\!7\!\rightarrow\!5$ are valid, giving $c_0[5]=2$ and, by $\binom{2}{2}=1$ butterfly.
\item At the remaining pivots ($u=5,7,2$), the priority rule prunes one or both endpoints, so no complete wedge pair is formed and each contributes a count of $0$.
\end{itemize}
A single \textsc{MPI\_Reduce} then sums the local counts,
$0+1+0+0=1$, yielding the correct total of one balanced butterfly.

This example confirms that the priority-based pruning enumerates each butterfly exactly once, even though its vertices are distributed across multiple processes.

\section{Experimental Evaluation}
\label{Sec:EE}
In this section, we assess the performance of the proposed algorithm.

\textbf{Computing Resources:} We implemented all algorithms in C++. Experiments were conducted on an HPC cluster using the \texttt{big\_compute\_amd\_9655} partition. Each compute node is equipped with two AMD EPYC 9655 processors, providing 192 physical CPU cores, 385~GB of main memory, and a 200~Gbps Mellanox HDR InfiniBand interconnect. The distributed implementation employs MPI for inter-node communication and Intel TBB for shared-memory parallelism. All algorithms, including the baseline methods, were evaluated under the same experimental settings.

\textbf{Algorithms:} We evaluate the following algorithms. 
\begin{itemize}
     \item \texttt{BB2K:} A serial algorithm from our prior work~\cite{kiran2024efficient}. 

    \item \texttt{M-BBC:} The proposed shared-memory multi-core algorithm for exact counting of balanced butterflies in large-scale signed graphs, exploiting fine-grained parallelism to accelerate butterfly counting~\cite{kiran2026multi}.
    
\item \texttt{S-Monarch:} An adaptation and extension of the distributed butterfly-counting algorithm~\cite{tang2024monarch} for balanced butterfly counting in signed bipartite graphs.

\item \texttt{D-BBC:} The proposed hybrid distributed-memory algorithm for counting balanced butterflies, employing MPI for inter-node communication and Intel oneTBB for intra-node parallelism.
\end{itemize}

\begin{table*}[!ht]
\centering
\footnotesize
\caption{Characteristics of the datasets with BBF (balanced butterflies count).}
\label{table:datasets}
\setlength{\tabcolsep}{4pt}
\renewcommand{\arraystretch}{1.2}

\begin{tabular}{|l|r|r|r|r|r|}
\hline
\textbf{Dataset} & $|U|$ & $|V|$ & $|E|$ & \textbf{Density} & \textbf{BBF} \\
\hline

Senate (\textbf{SE})       & 145 & 1,201 & 27,083 & 0.155 & 15.32 \textbf{M} \\
\hline
DBLP (\textbf{DBLP})       & 6,001 & 1,308 & 29,256 & 0.004 & 0.85 \textbf{M} \\
\hline
House (\textbf{HO})        & 515 & 1,281 & 114,378 & 0.174 & 280.79 \textbf{M} \\
\hline
Wiki-Nap (\textbf{NAP})    & 1,753 & 25,881 & 265,546 & 0.006 & 2.24 \textbf{B} \\
\hline
BookCrossing (\textbf{BC}) & 77,802 & 185,955 & 433,652 & 0.00003 & 1.11 \textbf{M} \\
\hline

NIPS-Papers (\textbf{NIPS}) & 1,500 & 12,375 & 746, 315 & 0.00003 & 3.75 \textbf{B} \\
\hline

Last.fm (\textbf{LF})      & 992 & 174,077 & 898,062 & 0.005 & 2.35 \textbf{B} \\
\hline
Movielens (\textbf{MV})    & 6,040 & 3,706 & 1,000,208 & 0.045 & 8.90 \textbf{B} \\
\hline
Jester 150 (\textbf{JE})   & 50,692 & 140 & 1,728,847 & 0.244 & 136.88 \textbf{B} \\
\hline
KDD Cup (\textbf{KDD})     & 255,170 & 1,848,114 & 2,766,393 & 0.0000059 & 9.20 \textbf{M} \\
\hline
Digg Votes (\textbf{DG})   & 139,409 & 3,553 & 3,010,197 & 0.006 & 15.06 \textbf{B} \\
\hline
AOL (\textbf{AOL})         & 4,811,647 & 1,632,788 & 10,741,953 & 0.0000014 & 104.47 \textbf{M} \\
\hline
Epinions (\textbf{EP})     & 120,492 & 755,760 & 13,668,320 & 0.00015 & 158.33 \textbf{B} \\
\hline
Netflix (\textbf{NX}) & 480,189 & 17,770 & 100,480,507 & 0.012 & 8.39 \textbf{T} \\
\hline
Yahoo (\textbf{YH})     & 1,000,990 & 624,961 & 256,804,235 & 0.0004 & 5.20 \textbf{T} \\
\hline

\end{tabular}
\end{table*}

\textbf{Datasets Description:}
We evaluate the proposed algorithms on a diverse collection of 15 real-world bipartite datasets, covering a wide range of graph sizes and densities. The SE and HO datasets are obtained from the signed bipartite network repository\footnote{\url{https://github.com/tylersnetwork/signed_bipartite_networks}}. The DBLP, MV, KDD, AOL, and NX datasets are collected from the bipartite network repository\footnote{\url{https://renchi.ac.cn/datasets/}}. The remaining datasets, namely NAP, BC, NIPS, LF, JE, DG, EP, and YH, are obtained from the KONECT collection\footnote{\url{http://konect.cc/networks/}}.

Among these datasets, SE, HO, BC, LF, and EP are naturally signed bipartite graphs and have been widely used in previous studies on balanced butterfly counting and maximal balanced biclique enumeration~\cite{derr2019balance,chung2023maximum,sun2022maximal}. The JE and MV datasets are rating networks that can be regarded as signed networks, as in prior work~\cite{sun2022maximal,kudo2020gcnext,chengunsupervised}. Following the convention adopted in these studies, ratings are converted into edge signs. The remaining datasets (DBLP, NAP, NIPS, KDD, AOL, DG, and YH) are originally unsigned bipartite graphs. Following the experimental setting of~\cite{li2018efficient}, we generate signed bipartite graphs by randomly assigning 30\% of the edges as negative and the remaining 70\% as positive. Several datasets, such as DG, NAP, and LF, contain duplicate edges with conflicting signs, which are resolved by retaining the latest interaction. A summary of all datasets is provided in Table~\ref{table:datasets} along with BBF (balanced butterflies count).

\subsection{Performance of the Proposed Algorithms}

The experimental evaluation consists of the following studies:

\begin{itemize}
    \item \textbf{Performance Comparison:} Comparison of the serial, shared-memory parallel, and distributed implementations on a single node.
    
    \item \textbf{Runtime Breakdown:} Analysis of the execution time by decomposing the total runtime into the I/O, Index, Exchange, Count, and Reduce phases.

    \item \textbf{Hybrid Configuration Analysis:} Evaluation of different MPI rank and TBB thread configurations to identify the best-performing hybrid execution setup.
    
    \item \textbf{Load Balancing:} Evaluation of the workload distribution among MPI processes.
    
    \item \textbf{Communication Analysis:} Analysis of communication volume with varying MPI ranks and across multiple compute nodes.
    
    \item \textbf{Multi-node Scalability Analysis:} Evaluation of count-phase performance and communication overhead on one, two, and four compute nodes.

\end{itemize}

We compare the proposed algorithm with the baseline methods in terms of execution time on the aforementioned datasets. The serial baseline (\texttt{BB2K}) is executed using a single CPU core, the shared-memory baseline (\texttt{M-BBC}) is evaluated on a single compute node with 192 threads, and the proposed distributed algorithm (\texttt{D-BBC}) is evaluated on 1, 2, and 4 compute nodes using multiple MPI processes and TBB thread configurations.

\begin{figure*}[t]
\centering
\begin{tikzpicture}
\begin{axis}[
    width=12cm,
    height=5cm,
    ybar=0.8pt,
    bar width=5pt,
    ymode=log,
    log basis y=10,
    log origin=infty,     
    ymin=0.0005,
    ymax=60000,
    symbolic x coords={SE,DBLP,HO,NAP,BC,NIPS,LF,MV,JE,KDD,DG,AOL,EP,NX,YH},
    xtick=data,
    xticklabel style={rotate=45,anchor=east,font=\footnotesize},
    xlabel={\footnotesize Datasets},
    ylabel={Execution Time (seconds)},
    ytick={0.001,0.01,0.1,1,10,100,1000},
    yticklabels={$10^{-3}$,$10^{-2}$,$10^{-1}$,$10^{0}$,$10^{1}$,$10^{2}$,\footnotesize $> 5\_hrs$},
    enlarge x limits=0.03,
    tick align=outside,
    axis on top,
    legend style={
        at={(0.5,1.10)},
        anchor=south,
        legend columns=3,
        draw=none,
        column sep=1.2ex,
        font=\small
    },
    every axis plot/.append style={fill opacity=0.9,draw opacity=1},
]

\addplot[
fill=blue!70,
draw=blue!90
]  coordinates {
(SE,0.03) (DBLP,0.05) (HO,0.92) (NAP,0.84) (BC,2.18)
(NIPS,30) (LF,12) (MV,42) (JE,12) (KDD,26)
(DG,153) (AOL,112) (EP,235) (NX,18000) (YH,18000)
};

\addplot[
fill=orange!80,
draw=orange!90!black
]  coordinates {
(SE,0.0069) (DBLP,0.0147) (HO,0.0147) (NAP,0.015) (BC,0.1261)
(NIPS,0.29) (LF,0.21) (MV,0.16) (JE,0.0919) (KDD,0.7289)
(DG,2.03) (AOL,2.91) (EP,4.0224) (NX,30.9) (YH,154.9508)
};

\addplot[
fill=green!70!black,
draw=green!50!black
] coordinates {
(SE,0.000950575) (DBLP,0.000758648) (HO,0.00191569) (NAP,0.00550389) (BC,0.00395155)
(NIPS,0.0136683) (LF,0.0146573) (MV,0.0166743) (JE,0.0249028) (KDD,0.0224538)
(DG,0.0410864) (AOL,0.289683) (EP,0.212017) (NX,3.23296) (YH,31.6319)
};

\legend{BB2K, M-BBC, D-BBC }
\end{axis}
\end{tikzpicture}
\caption{Counting-time comparison of BB2K, M-BBC, and D-BBC across all
datasets.}
\label{fig:runtimecomparison}
\end{figure*}
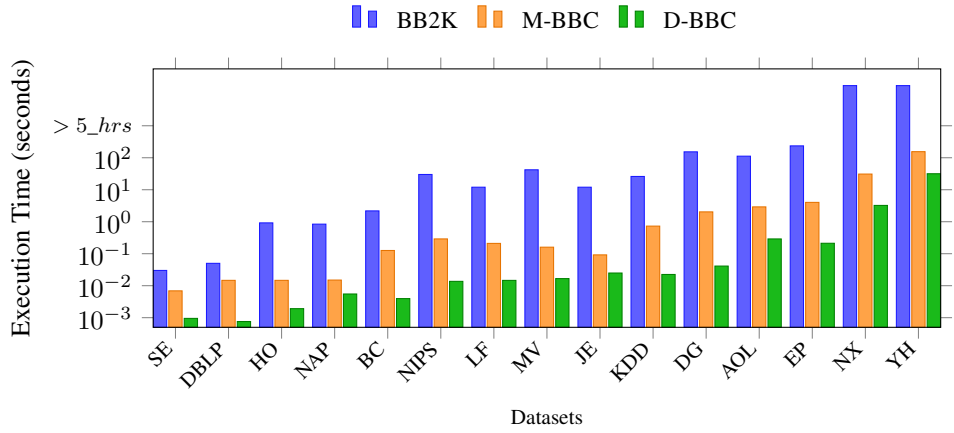

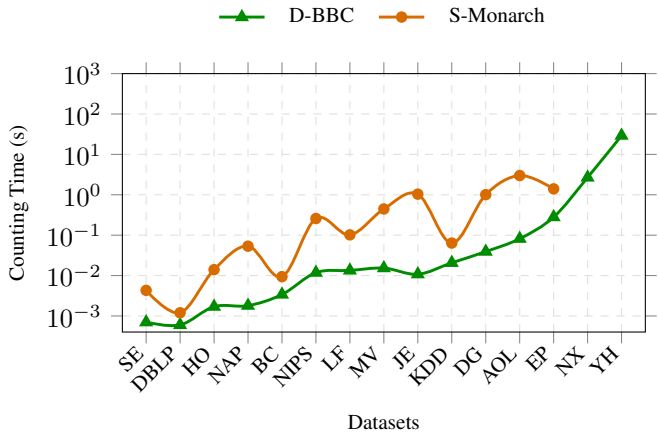
\begin{figure}[t]
\centering
\begin{tikzpicture}
\begin{axis}[
    width=8.5cm, height=5cm,
    ymode=log, log basis y=10,
    ymin=0.0004, ymax=1000,
    xtick=data,
    symbolic x coords={SE,DBLP,HO,NAP,BC,NIPS,LF,MV,JE,KDD,DG,AOL,EP,NX,YH},
    xticklabel style={rotate=45, anchor=east, font=\footnotesize},
    xlabel={\footnotesize Datasets},
    ylabel={\footnotesize Counting Time (s)},
    ytick={0.001,0.01,0.1,1,10,100,1000},
    yticklabels={$10^{-3}$,$10^{-2}$,$10^{-1}$,$10^{0}$,$10^{1}$,$10^{2}$,$10^{3}$},
    grid=major, grid style={dashed, gray!25},
    tick align=outside,
    legend style={
        at={(0.5,1.14)}, anchor=south,
        legend columns=2, draw=none, column sep=1.5ex, font=\footnotesize
    },
    enlarge x limits=0.05,
    every axis plot/.append style={line width=1.1pt},
]
\addplot[
    smooth, tension=0.6, color=green!55!black,
    mark=triangle*, mark size=2pt, mark options={fill=green!55!black},
] coordinates {
(SE,0.0007)(DBLP,0.0006)(HO,0.0017)(NAP,0.0018)(BC,0.0034)(NIPS,0.0118)(LF,0.0134)(MV,0.0153)(JE,0.0108)(KDD,0.0207)(DG,0.0392)(AOL,0.0817)(EP,0.2804)(NX,2.6827)(YH,29.1964)
};
\addlegendentry{D-BBC}

\addplot[
    smooth, tension=0.6, color=orange!85!black,
    mark=*, mark size=1.6pt, mark options={fill=orange!85!black},
] coordinates {
(SE,0.0043)(DBLP,0.0012)(HO,0.014)(NAP,0.0534)(BC,0.0094)(NIPS,0.2584)(LF,0.1018)(MV,0.4456)(JE,1.0368)(KDD,0.0638)(DG,1.007)(AOL,3.0)(EP,1.4019)
};
\addlegendentry{S-Monarch}
\end{axis}
\end{tikzpicture}
\caption{Counting-time comparison between \texttt{D-BBC} and the
\texttt{S-Monarch} baseline across all datasets (log scale). 
}
\label{fig:monarch_curve}
\end{figure}

\subsection{Performance Comparison}
Figure~\ref{fig:runtimecomparison} compares the counting time of the proposed \texttt{D-BBC} against the serial \texttt{BB2K} and the shared-memory \texttt{M-BBC} across all $15$ datasets. Across the $13$ datasets on which the serial baseline completed, \texttt{D-BBC} reduces the counting time by $31.6\times$ to $3723.9\times$ over
\texttt{BB2K}, with an average speedup of $540.9\times$. Relative to the shared-memory \texttt{M-BBC}, \texttt{D-BBC} achieves an additional $2.73\times$ (NAP) to $49.4\times$ (DG) improvement across all $15$ datasets, with an average speedup of $11.9\times$. For the two largest datasets, NX and YH, the serial \texttt{BB2K} did not finish within the $5$-hour execution limit. Both parallel methods, however, processed these graphs successfully: \texttt{D-BBC} reduces the counting time from $30.9$\,s to $3.23$\,s on NX ($9.6\times$) and from $154.95$\,s to $31.63$\,s on YH ($4.9\times$) relative to \texttt{M-BBC}, demonstrating its ability to scale to large-scale signed bipartite graphs that are intractable for a serial implementation.

\begin{table*}[t]
\centering
\caption{Configurations used in the experiments.}
\label{tab:configurations}
\renewcommand{\arraystretch}{1.1}
\setlength{\tabcolsep}{4pt}
\begin{tabular}{c|cccccccccccc}
\hline
\textbf{Configuration} &
$C_{1}$ & $C_{2}$ & $C_{3}$ & $C_{4}$ & $C_{5}$ & $C_{6}$ &
$C_{7}$ & $C_{8}$ & $C_{9}$ & $C_{10}$ & $C_{11}$ & $C_{12}$ \\
\hline
\textbf{MPI Ranks} &
1 & 2 & 3 & 4 & 6 & 8 & 12 & 16 & 24 & 32 & 48 & 64 \\

\textbf{TBB Threads} &
192 & 96 & 64 & 48 & 32 & 24 & 16 & 12 & 8 & 6 & 4 & 3 \\
\hline
\end{tabular}
\end{table*}

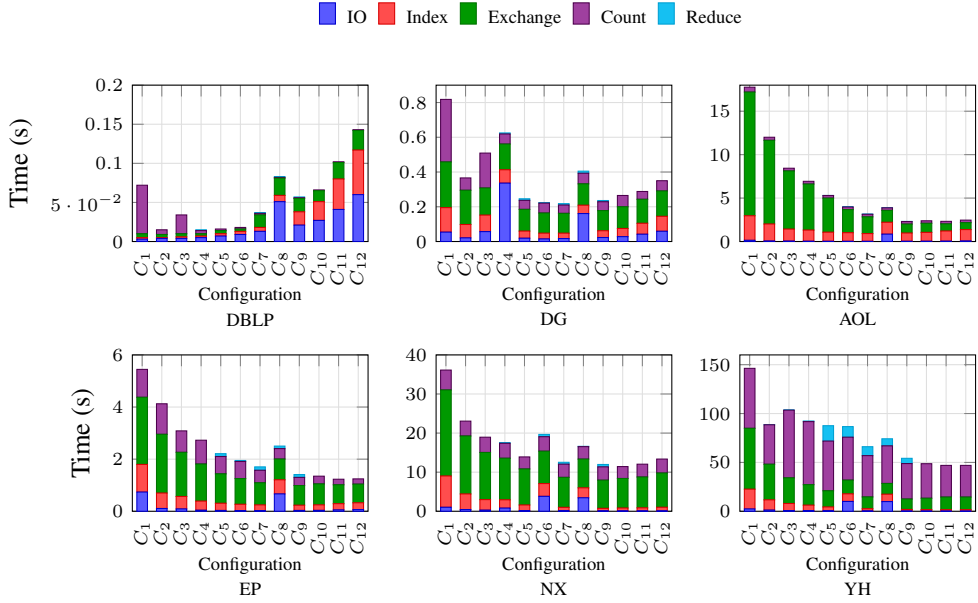
\begin{figure*}[!ht]

\centering
\begin{tikzpicture}
\begin{groupplot}[
    group style={group size=3 by 2, horizontal sep=0.9cm, vertical sep=1.5cm},
    width=4.7cm, height=3.65cm,
    ybar stacked,
    xmin=0.5, xmax=12.5,
    xtick={1,2,3,4,5,6,7,8,9,10,11,12},
    xticklabels={$C_{1}$,$C_{2}$,$C_{3}$,$C_{4}$,$C_{5}$,$C_{6}$,$C_{7}$,$C_{8}$,$C_{9}$,$C_{10}$,$C_{11}$,$C_{12}$},
    xticklabel style={rotate=90,anchor=east,font=\scriptsize},
    xlabel={Configuration},
xlabel style={font=\scriptsize, yshift=6pt},
    grid=major, grid style={gray!25},
    tick label style={font=\scriptsize},
    title style={font=\scriptsize},
]

\nextgroupplot[
    ymin=0, ymax=0.2,
    title={DBLP},    title style={yshift=-100pt}, ylabel={Time (s)},
    legend style={draw=none,fill=none,at={(1.7,1.30)},anchor=south,legend columns=5,column sep=0.5ex,font=\scriptsize}, 
]
\addplot+[ybar,bar width=4pt,fill=blue!65,draw=blue!80!black] coordinates{(1,0.003)(2,0.004)(3,0.004)(4,0.005)(5,0.007)(6,0.009)(7,0.013)(8,0.051)(9,0.021)(10,0.027)(11,0.041)(12,0.06)};
\addplot+[ybar,bar width=4pt,fill=red!70,draw=red!85!black] coordinates{(1,0.002)(2,0.001)(3,0.002)(4,0.002)(5,0.003)(6,0.004)(7,0.005)(8,0.008)(9,0.017)(10,0.024)(11,0.039)(12,0.057)};
\addplot+[ybar,bar width=4pt,fill=green!60!black,draw=green!45!black] coordinates{(1,0.005)(2,0.004)(3,0.004)(4,0.003)(5,0.004)(6,0.003)(7,0.016)(8,0.022)(9,0.017)(10,0.014)(11,0.021)(12,0.025)};
\addplot+[ybar,bar width=4pt,fill=violet!75,draw=violet!60!black] coordinates{(1,0.062)(2,0.006)(3,0.024)(4,0.004)(5,0.002)(6,0.002)(7,0.002)(8,0.001)(9,0.001)(10,0.001)(11,0.001)(12,0.001)};
\addplot+[ybar,bar width=4pt,fill=cyan!70,draw=cyan!85!black] coordinates{(1,0.0)(2,0.0)(3,0.0)(4,0.001)(5,0.0)(6,0.0)(7,0.001)(8,0.001)(9,0.001)(10,0.0)(11,0.0)(12,0.0)};
\legend{IO, Index, Exchange, Count, Reduce}

\nextgroupplot[
    ymin=0, ymax=0.9,
    title={DG},
    title style={yshift=-100pt},
]
\addplot+[ybar,bar width=4pt,fill=blue!65,draw=blue!80!black] coordinates{(1,0.054)(2,0.021)(3,0.057)(4,0.336)(5,0.019)(6,0.015)(7,0.017)(8,0.161)(9,0.023)(10,0.028)(11,0.042)(12,0.059)};
\addplot+[ybar,bar width=4pt,fill=red!70,draw=red!85!black] coordinates{(1,0.141)(2,0.077)(3,0.095)(4,0.077)(5,0.041)(6,0.033)(7,0.031)(8,0.048)(9,0.039)(10,0.047)(11,0.063)(12,0.086)};
\addplot+[ybar,bar width=4pt,fill=green!60!black,draw=green!45!black] coordinates{(1,0.263)(2,0.197)(3,0.156)(4,0.148)(5,0.125)(6,0.116)(7,0.114)(8,0.123)(9,0.116)(10,0.125)(11,0.138)(12,0.146)};
\addplot+[ybar,bar width=4pt,fill=violet!75,draw=violet!60!black] coordinates{(1,0.36)(2,0.071)(3,0.201)(4,0.058)(5,0.052)(6,0.057)(7,0.048)(8,0.061)(9,0.051)(10,0.065)(11,0.045)(12,0.059)};
\addplot+[ybar,bar width=4pt,fill=cyan!70,draw=cyan!85!black] coordinates{(1,0.0)(2,0.0)(3,0.0)(4,0.006)(5,0.009)(6,0.003)(7,0.008)(8,0.012)(9,0.007)(10,0.0)(11,0.0)(12,0.0)};

\nextgroupplot[
    ymin=0, ymax=18,    title style={yshift=-100pt},
    title={AOL},
]
\addplot+[ybar,bar width=4pt,fill=blue!65,draw=blue!80!black] coordinates{(1,0.14)(2,0.064)(3,0.054)(4,0.037)(5,0.029)(6,0.026)(7,0.025)(8,0.847)(9,0.031)(10,0.036)(11,0.046)(12,0.066)};
\addplot+[ybar,bar width=4pt,fill=red!70,draw=red!85!black] coordinates{(1,2.825)(2,1.955)(3,1.379)(4,1.282)(5,1.045)(6,0.991)(7,0.891)(8,1.368)(9,0.948)(10,1.016)(11,1.167)(12,1.309)};
\addplot+[ybar,bar width=4pt,fill=green!60!black,draw=green!45!black] coordinates{(1,14.237)(2,9.628)(3,6.692)(4,5.297)(5,3.92)(6,2.658)(7,1.923)(8,1.367)(9,1.026)(10,1.047)(11,0.817)(12,0.797)};
\addplot+[ybar,bar width=4pt,fill=violet!75,draw=violet!60!black] coordinates{(1,0.562)(2,0.356)(3,0.312)(4,0.315)(5,0.304)(6,0.3)(7,0.299)(8,0.296)(9,0.295)(10,0.288)(11,0.293)(12,0.291)};
\addplot+[ybar,bar width=4pt,fill=cyan!70,draw=cyan!85!black] coordinates{(1,0.0)(2,0.0)(3,0.0)(4,0.0)(5,0.001)(6,0.007)(7,0.003)(8,0.011)(9,0.016)(10,0.0)(11,0.0)(12,0.0)};
\nextgroupplot[
    ymin=0, ymax=6, 
    title={EP},    title style={yshift=-100pt},
    ylabel={Time (s)},
]
\addplot+[ybar,bar width=4pt,fill=blue!65,draw=blue!80!black] coordinates{(1,0.738)(2,0.107)(3,0.091)(4,0.041)(5,0.032)(6,0.027)(7,0.025)(8,0.667)(9,0.032)(10,0.035)(11,0.054)(12,0.065)};
\addplot+[ybar,bar width=4pt,fill=red!70,draw=red!85!black] coordinates{(1,1.06)(2,0.593)(3,0.477)(4,0.35)(5,0.276)(6,0.242)(7,0.22)(8,0.541)(9,0.201)(10,0.211)(11,0.238)(12,0.261)};
\addplot+[ybar,bar width=4pt,fill=green!60!black,draw=green!45!black] coordinates{(1,2.576)(2,2.253)(3,1.694)(4,1.423)(5,1.129)(6,0.986)(7,0.849)(8,0.797)(9,0.755)(10,0.805)(11,0.725)(12,0.715)};
\addplot+[ybar,bar width=4pt,fill=violet!75,draw=violet!60!black] coordinates{(1,1.071)(2,1.172)(3,0.823)(4,0.913)(5,0.665)(6,0.656)(7,0.48)(8,0.4)(9,0.312)(10,0.296)(11,0.213)(12,0.202)};
\addplot+[ybar,bar width=4pt,fill=cyan!70,draw=cyan!85!black] coordinates{(1,0.0)(2,0.0)(3,0.0)(4,0.0)(5,0.109)(6,0.038)(7,0.13)(8,0.097)(9,0.11)(10,0.0)(11,0.0)(12,0.0)};

\nextgroupplot[
    ymin=0, ymax=40,    title style={yshift=-100pt},
    title={NX},
]
\addplot+[ybar,bar width=4pt,fill=blue!65,draw=blue!80!black] coordinates{(1,1.023)(2,0.441)(3,0.293)(4,0.812)(5,0.15)(6,3.818)(7,0.112)(8,3.451)(9,0.088)(10,0.078)(11,0.079)(12,0.084)};
\addplot+[ybar,bar width=4pt,fill=red!70,draw=red!85!black] coordinates{(1,8.011)(2,3.992)(3,2.686)(4,2.12)(5,1.458)(6,3.268)(7,0.889)(8,2.546)(9,0.654)(10,0.756)(11,0.8)(12,0.942)};
\addplot+[ybar,bar width=4pt,fill=green!60!black,draw=green!45!black] coordinates{(1,21.997)(2,14.824)(3,12.063)(4,10.634)(5,9.212)(6,8.28)(7,7.67)(8,7.344)(9,7.205)(10,7.552)(11,7.888)(12,8.762)};
\addplot+[ybar,bar width=4pt,fill=violet!75,draw=violet!60!black] coordinates{(1,5.089)(2,3.804)(3,3.906)(4,3.824)(5,3.084)(6,3.71)(7,3.376)(8,3.193)(9,3.418)(10,3.05)(11,3.278)(12,3.568)};
\addplot+[ybar,bar width=4pt,fill=cyan!70,draw=cyan!85!black] coordinates{(1,0.0)(2,0.0)(3,0.0)(4,0.153)(5,0.0)(6,0.583)(7,0.481)(8,0.076)(9,0.609)(10,0.0)(11,0.0)(12,0.0)};

\nextgroupplot[
    ymin=0, ymax=160,    title style={yshift=-100pt},
    title={YH},
]
\addplot+[ybar,bar width=4pt,fill=blue!65,draw=blue!80!black] coordinates{(1,2.387)(2,1.233)(3,0.723)(4,0.582)(5,0.433)(6,10.06)(7,0.292)(8,9.932)(9,0.211)(10,0.18)(11,0.138)(12,0.138)};
\addplot+[ybar,bar width=4pt,fill=red!70,draw=red!85!black] coordinates{(1,20.258)(2,10.462)(3,7.257)(4,5.702)(5,3.999)(6,7.796)(7,2.307)(8,7.535)(9,1.637)(10,1.536)(11,1.483)(12,1.483)};
\addplot+[ybar,bar width=4pt,fill=green!60!black,draw=green!45!black] coordinates{(1,62.25)(2,36.433)(3,26.244)(4,20.948)(5,16.535)(6,14.012)(7,12.087)(8,11.008)(9,10.698)(10,11.627)(11,12.92)(12,12.92)};
\addplot+[ybar,bar width=4pt,fill=violet!75,draw=violet!60!black] coordinates{(1,61.36)(2,40.392)(3,69.268)(4,64.53)(5,50.857)(6,43.942)(7,42.263)(8,38.492)(9,36.287)(10,35.346)(11,32.415)(12,32.415)};
\addplot+[ybar,bar width=4pt,fill=cyan!70,draw=cyan!85!black] coordinates{(1,0.0)(2,0.002)(3,0.004)(4,0.611)(5,15.753)(6,10.839)(7,9.134)(8,7.215)(9,5.363)(10,0.0)(11,0.0)(12,0.0)};

\end{groupplot}
\end{tikzpicture}
\caption{Per-phase execution-time breakdown across configurations.
}
\label{fig:phase_breakdown}
\end{figure*}

\begin{figure*}[t]
\centering
\begin{tikzpicture}
\begin{axis}[
    width=12cm, height=5.0cm,
    ybar=1pt,
    bar width=6pt,
    ymode=log, log basis y=10,
    ymin=0.01, ymax=20,
    log origin=infty,
    symbolic x coords={SE,DBLP,HO,NAP,BC,NIPS,LF,MV,JE,KDD,DG,AOL, EP},
    xtick=data,
    xticklabel style={rotate=45, anchor=east, font=\footnotesize},
    xlabel={\footnotesize Datasets},
    ylabel={\footnotesize End-to-End Time (s)},
    ytick={0.01,0.1,1,10},
    yticklabels={$10^{-2}$,$10^{-1}$,$10^{0}$,$10^{1}$},
    ymajorgrids=true, grid style={dashed, gray!25},
    tick align=outside,
    legend style={
        at={(0.5,1.14)}, anchor=south,
        legend columns=2, draw=none, column sep=1.5ex, font=\small
    },
    enlarge x limits=0.05,
    every axis plot/.append style={fill opacity=0.9, draw opacity=1},
]
\addplot[fill=green!60!black, draw=green!45!black] coordinates {
(SE,0.0230)(DBLP,0.0414)(HO,0.0487)(NAP,0.0672)(BC,0.0963)(NIPS,0.0985)(LF,0.1312)(MV,0.1057)(JE,0.1289)(KDD,0.5541)(DG,0.2148)(AOL, 1.2654) (EP,1.2795)
};
\addlegendentry{D-BBC}

\addplot[fill=orange!80, draw=orange!90!black] coordinates {
(SE,0.0281)(DBLP,0.0330)(HO,0.1224)(NAP,0.2477)(BC,0.1927)(NIPS,0.5822)(LF,0.4296)(MV,0.9009)(JE,3.0401)(KDD,3.8488)(DG,3.2107) (AOL, 6.2654)(EP,5.9708)
};
\addlegendentry{S-Monarch}
\end{axis}
\end{tikzpicture}
\caption{End-to-end execution time of \texttt{D-BBC} andthe \texttt{S-Monarch} baseline across all datasets.}
\label{fig:e2e_bar}
\end{figure*}

Figure~\ref{fig:monarch_curve} compares the counting time of \texttt{D-BBC} with the distributed \texttt{S-Monarch} baseline across all datasets. Among the 13 datasets completed by \texttt{S-Monarch}, \texttt{D-BBC} consistently outperforms the baseline, achieving speedups ranging from $2.00\times$ (DBLP) to $96.00\times$ (JE), with an average speedup of $21.07\times$. For datasets with relatively few balanced butterflies, such as DBLP, BC, and KDD, the performance gap is small because the overhead of explicit balanced-butterfly verification is limited. As illustrated in Figure~\ref{fig:monarch_curve}, the performance gap becomes more pronounced as the graph size increases. In particular, \texttt{S-Monarch} fails to complete on the NX and YH datasets due to their substantially higher computational cost and memory requirements. In contrast, \texttt{D-BBC} successfully processes NX and YH in $2.68,\mathrm{s}$ and $29.20,\mathrm{s}$, respectively, demonstrating its superior scalability on large graphs. The improved performance of \texttt{D-BBC} is primarily due to the elimination of the explicit balanced-butterfly verification required by the adapted \texttt{S-Monarch} baseline.

\subsection{Runtime Breakdown and Scaling Analysis}

To understand how execution time is distributed across the proposed algorithm, we decompose the end-to-end runtime into five phases: \emph{IO} (parallel graph ingestion), \emph{Index} (global vertex indexing, degree computation, and pivot assignment), \emph{Exchange} (all-to-all wedge/edge communication), \emph{Count} (local butterfly counting), and \emph{Reduce} (global aggregation of local butterfly counts).

Table~\ref{tab:configurations} summarizes the hybrid MPI+TBB configurations evaluated in the experiments. Each configuration uses 192 CPU cores, varying the number of MPI ranks and TBB threads. Figure~\ref{fig:phase_breakdown} presents the runtime breakdown for six representative datasets, including two large (YH and NX), two medium (AOL and EP), and two small (DG and DBLP) graphs, across the twelve hybrid configurations ($C_{1}$--$C_{12}$). This analysis highlights how the execution time of each phase evolves with different MPI/TBB configurations and graph sizes. As shown in Figure~\ref{fig:phase_breakdown}, the \emph{Count} phase dominates the execution time for large datasets, whereas its contribution decreases for medium and small datasets. Increasing the number of MPI ranks significantly reduces the counting time but increases the \emph{Exchange} time because of additional inter-process communication. For large graphs, the reduction in computation outweighs the communication overhead, leading to the best overall performance. In contrast, for medium and small graphs, communication overhead becomes increasingly significant, resulting in diminishing performance gains. Across all datasets, the \emph{IO}, \emph{Index}, and \emph{Reduce} phases remain relatively small.

Figure~\ref{fig:e2e_bar} compares the end-to-end execution time (including I/O, Communication, and Counting phase) of \texttt{D-BBC} and \texttt{S-Monarch} across all datasets. \texttt{D-BBC} consistently achieves lower end-to-end execution times, confirming the effectiveness of the proposed approach.

\subsection{Hybrid Configuration Analysis}

Figure~\ref{fig:speedup_heavy} presents the speedup of the six largest datasets under different hybrid MPI+TBB configurations, using the single-rank configuration ($C_1$) as the baseline. The speedup generally increases with the number of MPI ranks, demonstrating the effectiveness of the proposed hybrid parallelization strategy. The largest speedups are observed for the computation-intensive datasets, where additional parallelism significantly reduces butterfly-counting time. At higher MPI rank counts, the performance gain gradually saturates as communication overhead increasingly offsets the reduction in computation time. Overall, configurations $C_9$--$C_{12}$ consistently deliver the highest speedups, indicating an effective balance between inter-process communication and intra-process thread parallelism.

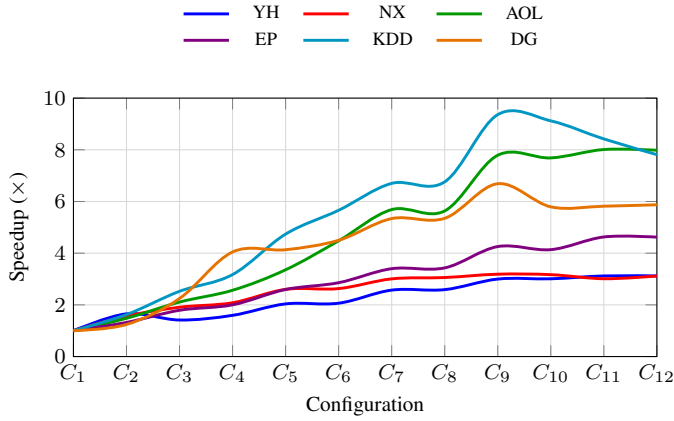
\begin{figure}[t]
\centering
\footnotesize
\begin{tikzpicture}
\begin{axis}[
    width=9.3cm, height=5.0cm,
    xmin=1, xmax=12,
    ymin=0, ymax=10,
    xtick={1,2,3,4,5,6,7,8,9,10,11,12},
    xticklabels={%
        {$C_{1}$},{$C_{2}$},{$C_{3}$},{$C_{4}$},%
        {$C_{5}$},{$C_{6}$},{$C_{7}$},{$C_{8}$},%
        {$C_{9}$},{$C_{10}$},{$C_{11}$},{$C_{12}$}},
    xticklabel style={anchor=north},
    xlabel={Configuration},
    ylabel={Speedup ($\times$)},
    grid=major, grid style={gray!30},
    every axis plot/.append style={line width=1.1pt, smooth, tension=0.6},
 legend style={
    draw=none,
    fill=none,
    at={(0.5,1.15)},
    anchor=south,
    legend columns=3,
    column sep=0.8em,
    font=\scriptsize
},
]
\addplot[blue]
coordinates{(1,1.000)(2,1.653)(3,1.414)(4,1.595)(5,2.038)(6,2.065)(7,2.571)(8,2.590)(9,2.999)(10,3.009)(11,3.120)(12,3.120)};
\addlegendentry{YH}
\addplot[red]
coordinates{(1,1.000)(2,1.566)(3,1.908)(4,2.082)(5,2.600)(6,2.630)(7,3.006)(8,3.058)(9,3.187)(10,3.169)(11,3.011)(12,3.117)};
\addlegendentry{NX}
\addplot[green!60!black]
coordinates{(1,1.000)(2,1.480)(3,2.109)(4,2.565)(5,3.357)(6,4.479)(7,5.687)(8,5.636)(9,7.790)(10,7.687)(11,8.011)(12,7.983)};
\addlegendentry{AOL}
\addplot[violet]
coordinates{(1,1.000)(2,1.320)(3,1.791)(4,2.001)(5,2.595)(6,2.858)(7,3.400)(8,3.431)(9,4.256)(10,4.137)(11,4.628)(12,4.620)};
\addlegendentry{EP}
\addplot[cyan!75!black]
coordinates{(1,1.000)(2,1.624)(3,2.525)(4,3.172)(5,4.738)(6,5.660)(7,6.701)(8,6.761)(9,9.362)(10,9.120)(11,8.421)(12,7.813)};
\addlegendentry{KDD}
\addplot[orange!90!black]
coordinates{(1,1.000)(2,1.241)(3,2.247)(4,4.047)(5,4.134)(6,4.499)(7,5.339)(8,5.35)(9,6.688)(10,5.793)(11,5.818)(12,5.874)};
\addlegendentry{DG}
\end{axis}
\end{tikzpicture}
\caption{Speedup for different hybrid MPI+TBB configurations, with $C_1$ as the baseline.}
\label{fig:speedup_heavy}
\end{figure}

\begin{table*}[t]
\centering
\caption{Single-node dynamic load imbalance across all configurations.}
\label{tab:imbalance}
\setlength{\tabcolsep}{4pt}
\footnotesize
\begin{tabular}{l*{12}{r}}
\toprule
 & \multicolumn{12}{c}{Configuration (MPI ranks $\times$ threads/rank)} \\
\cmidrule(lr){2-13}
Dataset
 & \shortstack{$C_{1}$\\\tiny$1{\times}192$}
 & \shortstack{$C_{2}$\\\tiny$2{\times}96$}
 & \shortstack{$C_{3}$\\\tiny$3{\times}64$}
 & \shortstack{$C_{4}$\\\tiny$4{\times}48$}
 & \shortstack{$C_{5}$\\\tiny$6{\times}32$}
 & \shortstack{$C_{6}$\\\tiny$8{\times}24$}
 & \shortstack{$C_{7}$\\\tiny$12{\times}16$}
 & \shortstack{$C_{8}$\\\tiny$16{\times}12$}
 & \shortstack{$C_{9}$\\\tiny$24{\times}8$}
 & \shortstack{$C_{10}$\\\tiny$32{\times}6$}
 & \shortstack{$C_{11}$\\\tiny$48{\times}4$}
 & \shortstack{$C_{12}$\\\tiny$64{\times}3$} \\
\midrule
YH   & 1.00 & 1.11 & 1.84 & 1.68 & 1.70 & 1.43 & 1.54 & 1.39 & 1.34 & 1.27 & 1.19 & 1.19 \\

NX   & 1.00 & 1.02 & 1.39 & 1.22 & 1.21 & 
1.32 & 1.35 & 1.25 & 1.56 & 1.29 & 1.66 & 1.84 \\

AOL  & 1.00 & 1.17 & 1.41 & 1.51 & 1.94 & 2.02 & 2.42 & 2.40 & 2.54 & 2.64 & 2.67 & 2.89 \\

EP   & 1.00 & 1.03 & 1.07 & 1.22 & 1.46 & 1.55 & 1.66 & 1.69 & 1.82 & 1.95 & 2.04 & 2.28 \\

KDD  & 1.00 & 1.03 & 1.34 & 1.10 & 1.06 & 1.15 & 1.23 & 1.10 & 1.24 & 1.56 & 1.77 & 1.71 \\

DG   & 1.00 & 1.16 & 1.02 & 1.12 & 1.24 & 1.45 & 1.37 & 1.41 & 1.44 & 1.59 & 1.51 & 1.89 \\

JE   & 1.00 & 1.05 & 1.76 & 1.35 & 1.81 & 1.64 & 1.93 & 1.62 & 1.76 & 1.72 & 1.92 & 1.93 \\

LF   & 1.00 & 1.10 & 1.38 & 1.40 & 1.53 & 1.56 & 1.67 & 1.44 & 1.50 & 1.82 & 1.75 & 1.66 \\

BC   & 1.00 & 1.05 & 1.55 & 1.25 & 1.09 & 1.22 & 1.50 & 1.44 & 1.56 & 1.56 & 1.73 & 1.40 \\

MV   & 1.00 & 1.78 & 4.20 & 1.24 & 1.47 & 1.68 & 1.43 & 1.87 & 1.57 & 1.62 & 1.58 & 1.66 \\

NIPS & 1.00 & 1.28 & 1.80 & 1.71 & 1.87 & 1.98 & 1.39 & 1.82 & 1.57 & 1.46 & 1.74 & 1.65 \\

NAP  & 1.00 & 1.07 & 1.94 & 1.13 & 1.66 & 2.10 & 2.14 & 2.16 & 2.54 & 2.52 & 2.71 & 2.91 \\

HO   & 1.00 & 1.18 & 1.94 & 1.90 & 1.51 & 1.72 & 1.42 & 2.02 & 2.15 & 2.33 & 2.72 & 2.93 \\

DBLP & 1.00 & 1.29 & 1.48 & 1.48 & 1.60 & 1.20 & 1.35 & 1.73 & 2.05 & 2.43 & 2.55 & 2.66 \\

SE   & 1.00 & 1.07 & 1.19 & 1.21 & 1.21 & 1.46 & 1.57 & 1.75 & 2.02 & 2.29 & 2.59 & 2.9 \\
\bottomrule
\end{tabular}
\end{table*}

\subsection{Load Balancing Analysis}

Table~\ref{tab:imbalance} reports the dynamic load imbalance, measured by the straggler ratio ($\beta=t_{\max}^{\mathrm{count}}/t_{\min}^{\mathrm{count}}$), across all 15 datasets and 12 hybrid configurations. Since the proposed partitioning balances the estimated workload, the remaining imbalance reflects variations in the actual cost of counting during execution.

Overall, the proposed strategy achieves good dynamic load balancing, particularly for the large, computation-intensive datasets such as YH and NX, where the straggler ratio remains relatively low across most configurations. In contrast, smaller datasets, including DBLP, SE, NAP, and HO, exhibit higher imbalance under fine-grained configurations because their limited computational workload amplifies execution-time variations across MPI ranks. However, these datasets have very short execution times, and therefore the increased imbalance has a negligible impact on the overall runtime. These results demonstrate that the proposed partitioning strategy effectively balances the workloads of large graphs, where load balancing has the greatest impact on performance, while maintaining acceptable imbalance across the remaining datasets.

\subsection{Communication Analysis}
Figure~\ref{fig:bytes_sent} illustrates the total communication volume (bytes sent) for six representative datasets as the number of MPI ranks increases while maintaining a fixed worker budget of 192 CPU cores. For all datasets, the communication volume increases monotonically with the number of MPI ranks. For example, the communication volume of YH grows from 11.0 GB with a single MPI rank to 193.8 GB with 48 MPI ranks, while NX increases from 4.32 GB to 100.36 GB. The increase in communication volume is expected because partitioning the graph among more MPI processes creates additional boundary vertices and wedges that must be exchanged between processes. Although the total amount of communicated data increases, the runtime of the communication phase on a single node decreases due to efficient shared-memory communication and reduced per-rank computational workload. Consequently, the additional communication volume does not become a performance bottleneck in the single-node experiments.


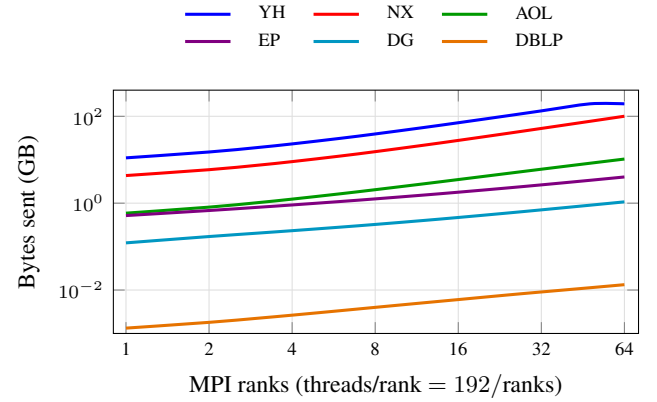
\begin{figure}[t]
\centering
\begin{tikzpicture}
\begin{axis}[
    width=\linewidth,
    height=4.8cm,
    xmode=log,
    log basis x=2,
    ymode=log,
    log basis y=10,
    xmin=0.9,
    xmax=72,
    xtick={1,2,4,8,16,32,64},
    xticklabels={1,2,4,8,16,32,64},
    ymin=0.001,
    ymax=400,
    xlabel={MPI ranks (threads/rank $=192/\text{ranks}$)},
    ylabel={Bytes sent (GB)},
    xlabel style={font=\small},
    ylabel style={font=\small},
    tick label style={font=\scriptsize},
    grid=both,
    major grid style={gray!25},
    minor grid style={gray!12},
    legend style={
        font=\scriptsize,
        at={(0.5,1.12)},
        anchor=south,
        draw=none,
        fill=none,
        legend columns=3,
        column sep=0.8em
    },
    legend cell align=left,
    every axis plot/.append style={
        line width=1.15pt,
        smooth,
    },
]

\addplot[color=blue]
coordinates {
    (1,11.046287)
    (2,15.077907)
    (3,19.109486)
    (4,23.140332)
    (6,31.193123)
    (8,39.231913)
    (12,55.218021)
    (16,71.092951)
    (24,102.471218)
    (32,133.317683)
    (48,193.845063)
    (64,193.845063)
};

\addplot[color=red]
coordinates {
    (1,4.323149)
    (2,5.898894)
    (3,7.474444)
    (4,9.049625)
    (6,12.197512)
    (8,15.342111)
    (12,21.613955)
    (16,27.856455)
    (24,40.261048)
    (32,52.558208)
    (48,76.702617)
    (64,100.361473)
};

\addplot[color=green!60!black]
coordinates {
    (1,0.588722)
    (2,0.81022)
    (3,1.027397)
    (4,1.239383)
    (6,1.648939)
    (8,2.042353)
    (12,2.789847)
    (16,3.496991)
    (24,4.819119)
    (32,6.050057)
    (48,8.304234)
    (64,10.36988)
};

\addplot[color=violet]
coordinates {
    (1,0.517954)
    (2,0.677756)
    (3,0.802673)
    (4,0.909803)
    (6,1.093826)
    (8,1.254863)
    (12,1.537528)
    (16,1.788203)
    (24,2.235324)
    (32,2.638026)
    (48,3.36019)
    (64,4.012073)
};

\addplot[color=cyan!75!black]
coordinates {
    (1,0.121858)
    (2,0.170499)
    (3,0.204326)
    (4,0.232565)
    (6,0.281115)
    (8,0.323871)
    (12,0.399902)
    (16,0.468274)
    (24,0.590851)
    (32,0.701341)
    (48,0.897226)
    (64,1.071956)
};

\addplot[color=orange!90!black]
coordinates {
    (1,0.001327)
    (2,0.001803)
    (3,0.002241)
    (4,0.002638)
    (6,0.003355)
    (8,0.003992)
    (12,0.005086)
    (16,0.006027)
    (24,0.007669)
    (32,0.009026)
    (48,0.011291)
    (64,0.013283)
};

\legend{YH,NX,AOL,EP,DG,DBLP}

\end{axis}
\end{tikzpicture}

\caption{Communication volume (total bytes sent, summed over all MPI ranks)
as a function of the number of MPI ranks for six representative datasets.}
\label{fig:bytes_sent}
\end{figure}

\begin{figure*}[t]
\centering
\begin{minipage}{0.49\linewidth}
\centering
\begin{tikzpicture}
\begin{axis}[
    width=9cm, height=5.0cm,
    ybar=2pt, bar width=7pt,
    ymode=log, log basis y=10,
    log origin y=infty,
    symbolic x coords={DG,KDD,EP,AOL,NX,YH},
    xtick=data,
    xticklabel style={font=\scriptsize},
    ymin=0.01, ymax=200,
    ytick={0.01,0.1,1,10,100},
    yticklabels={$10^{-2}$,$10^{-1}$,$10^{0}$,$10^{1}$,$10^{2}$},
    ylabel={Count-phase time (s)},
    ylabel style={font=\small},
    tick label style={font=\scriptsize},
    grid=major, grid style={gray!15},
    enlarge x limits=0.14,
    legend style={font=\scriptsize, at={(0.5,1.03)}, anchor=south,
                  draw=none, fill=none,
                  /tikz/every even column/.append style={column sep=6pt}},
    legend columns=-1,
]
\addplot+[fill=blue!35, draw=blue!65!black]
    coordinates {(DG,0.091)(KDD,0.093)(EP,0.712)(AOL,0.695)(NX,3.918)(YH,46.287)};
\addplot+[fill=orange!70,draw=orange!85!black]
    coordinates {(DG,0.040)(KDD,0.046)(EP,0.416)(AOL,0.335)(NX,1.948)(YH,20.706)};
\addplot+[fill=green!45!black,draw=green!30!black]
    coordinates {(DG,0.025)(KDD,0.020)(EP,0.150)(AOL,0.100)(NX,0.500)(YH,9.000)};
\legend{1 node, 2 nodes, 4 nodes}
\end{axis}
\end{tikzpicture}
\subcaption{Count-phase time.}
\label{fig:count_1n_2n_4n_top6}
\end{minipage}
\hfill
\begin{minipage}{0.49\linewidth}
\centering
\begin{tikzpicture}
\begin{axis}[
    width=9cm, height=5.0cm,
    ybar=2pt, bar width=7pt,
    ymode=log, log basis y=10,
    log origin y=infty,
    symbolic x coords={DG,KDD,EP,AOL,NX,YH},
    xtick=data,
    xticklabel style={font=\scriptsize},
    ymin=0.2, ymax=1500,
    ytick={1,10,100,1000},
    yticklabels={$10^{0}$,$10^{1}$,$10^{2}$,$10^{3}$},
    ylabel={Bytes sent (GB)},
    ylabel style={font=\small},
    tick label style={font=\scriptsize},
    grid=major, grid style={gray!15},
    enlarge x limits=0.14,
    legend style={font=\scriptsize, at={(0.5,1.03)}, anchor=south,
                  draw=none, fill=none,
                  /tikz/every even column/.append style={column sep=6pt}},
    legend columns=-1,
]
\addplot+[fill=blue!35,draw=blue!65!black]
    coordinates {(DG,1.131)(KDD,0.591)(EP,4.819)(AOL,2.235)(NX,40.261)(YH,102.471)};
\addplot+[fill=orange!70,draw=orange!85!black]
    coordinates {(DG,2.065)(KDD,0.897)(EP,8.304)(AOL,3.360)(NX,76.703)(YH,193.845)};
\addplot+[pattern=north east lines, fill=green!45!black,draw=green!30!black]
    coordinates {(DG,3.77)(KDD,1.36)(EP,14.31)(AOL,5.05)(NX,146.1)(YH,366.7)};
\legend{1 node, 2 nodes, 4 nodes}
\end{axis}
\end{tikzpicture}
\subcaption{Communication volume (bytes sent).}
\label{fig:bytes_1n_2n_4n_top6}
\end{minipage}
\caption{Multi-node behavior on the six largest datasets.}
\label{fig:multinode_count_bytes}
\end{figure*}
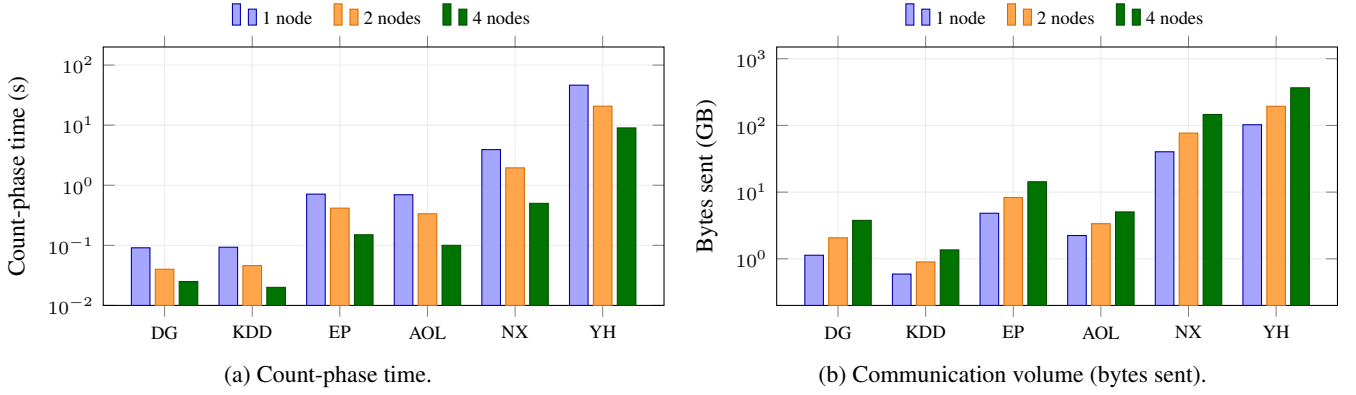

\begin{figure*}[t]
\centering

\begin{subfigure}[t]{0.48\textwidth}
\centering
\begin{tikzpicture}
\begin{axis}[
    width=9cm,
    height=5.0cm,
    xlabel={Dataset},
    ybar=2pt,
    bar width=7pt,
    ymode=log,
    log basis y=10,
    log origin y=infty,
    symbolic x coords={DG,KDD,EP,AOL,NX,YH},
    xtick=data,
    xticklabel style={font=\scriptsize},
    ymin=0.1,
    ymax=100,
    ytick={0.1,1,10,100},
    yticklabels={$10^{-1}$,$10^{0}$,$10^{1}$,$10^{2}$},
    ylabel={Exchange time (s)},
    ylabel style={font=\scriptsize},
    tick label style={font=\scriptsize},
    grid=major,
    grid style={gray!15},
    enlarge x limits=0.14,
    legend style={
        font=\scriptsize,
        at={(0.5,1.03)},
        anchor=south,
        draw=none,
        fill=none
    },
    legend columns=-1,
]

\addplot+[fill=blue!35,draw=blue!70!black]
coordinates{
(DG,0.569)
(KDD,0.507)
(EP,2.604)
(AOL,2.011)
(NX,20.609)
(YH,51.701)
};

\addplot+[fill=orange!70,draw=orange!85!black]
coordinates{
(DG,0.458)
(KDD,0.380)
(EP,2.006)
(AOL,1.501)
(NX,15.471)
(YH,38.963)
};

\addplot+[fill=green!45!black,draw=green!30!black]
coordinates{
(DG,0.452)
(KDD,0.318)
(EP,1.829)
(AOL,1.304)
(NX,12.831)
(YH,32.587)
};

\legend{1 node,2 nodes,4 nodes}

\end{axis}
\end{tikzpicture}
\caption{Exchange time.}
\label{fig:exchange_best}
\end{subfigure}
\hfill
\begin{subfigure}[t]{0.48\textwidth}
\centering
\begin{tikzpicture}
\begin{axis}[
    width=9.0cm,
    height=5.0cm,
    xlabel={Dataset},
    ylabel={Peak RSS per rank (MB)},
    symbolic x coords={DG,KDD,EP,AOL,NX,YH},
    xtick=data,
    ymode=log,
    log basis y=10,
    ymin=100,
    ymax=100000,
    ytick={100,1000,10000,100000},
    yticklabels={$10^{2}$,$10^{3}$,$10^{4}$,$10^{5}$},
    grid=major,
    grid style={gray!20},
    tick label style={font=\scriptsize},
    label style={font=\scriptsize},
    legend style={
        font=\scriptsize,
        at={(0.5,1.03)},
        anchor=south,
        draw=none,
        fill=none
    },
    legend columns=3,
]

\addplot[
blue,
very thick,
smooth,
tension=0.4,
]
coordinates{
(DG,604)
(KDD,366)
(EP,2604)
(AOL,1298)
(NX,25999)
(YH,55589)
};

\addplot[
red,
very thick,
smooth,
tension=0.4,
]
coordinates{
(DG,565)
(KDD,295)
(EP,2130)
(AOL,1128)
(NX,17040)
(YH,48115)
};

\addplot[
green!60!black,
very thick,
smooth,
tension=0.4,
]
coordinates{
(DG,514)
(KDD,247)
(EP,1888)
(AOL,964)
(NX,17053)
(YH,40106)
};

\legend{1 node,2 nodes,4 nodes}

\end{axis}
\end{tikzpicture}
\caption{Peak memory usage.}
\label{fig:memory_usage}
\end{subfigure}

\caption{Communication and memory analysis of the proposed distributed algorithm. (a) Exchange time decreases as the number of nodes increases due to improved communication parallelism. (b) Peak memory usage per MPI rank decreases as the graph is partitioned across more processes.}
\label{fig:comm_memory}

\end{figure*}
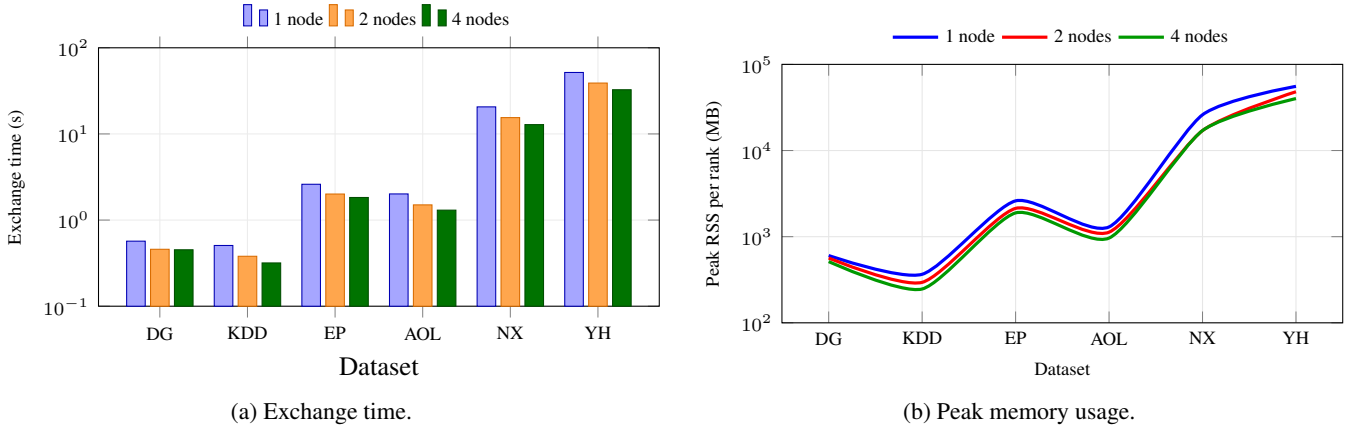



\subsection{Multi-node Scalability Analysis}

Figure~\ref{fig:count_1n_2n_4n_top6} shows the count-phase execution time of the proposed distributed algorithm on one, two, and four compute nodes. The execution time consistently decreases as the number of nodes increases, indicating the strong scalability of the proposed approach. Larger datasets benefit more from distributed execution due to their higher computational workload, whereas smaller datasets achieve relatively lower improvements.

Figure~\ref{fig:bytes_1n_2n_4n_top6} presents the communication volume (bytes sent) of the proposed distributed algorithm on one, two, and four compute nodes. As the deployment scales out, the communication volume increases because finer graph partitioning across more MPI processes introduces additional boundary vertices and wedges that require inter-process communication. Despite this increase, the counting phase continues to scale effectively across all datasets, with execution time consistently decreasing as more compute nodes are employed.

Figure~\ref{fig:exchange_best} presents the communication (exchange) time achieved using the best MPI/thread configuration for each node count. As the number of MPI ranks increases, the communication volume increases because the pivot exchange involves more graph partitions, leading to additional cross-rank communication. Consequently, the total number of bytes exchanged is primarily determined by the number of physical nodes. Despite the increase in communication volume, communication time decreases as execution scales from 1 to 4 nodes. This improvement is achieved by distributing the communication workload across more MPI processes and network links, allowing each process to exchange a smaller portion of the data in parallel. Overall, these results demonstrate that the communication phase scales efficiently and does not become the primary performance bottleneck as the distributed execution scales across multiple nodes.

Figure~\ref{fig:memory_usage} shows the peak memory usage (RSS) per MPI rank for the best MPI/thread configuration at each node count. As the number of nodes and MPI ranks increases, the memory usage per rank consistently decreases because the input graph is partitioned across more MPI processes, allowing each process to store only a smaller portion of the graph. For example, peak memory usage for the NX dataset decreases from approximately 25.9 GB to 17 GB, and for the YH dataset from about 55 GB to 40 GB. These results demonstrate that the proposed distributed approach effectively reduces per-rank memory usage, enabling the processing of larger graphs that may not fit within the memory capacity of a single compute node.

\section{Ablation Study:}

\subsection{Pure MPI vs. Hybrid MPI+TBB}

To evaluate the contribution of the hybrid programming model, we compare \texttt{D-BBC} using $24$ MPI ranks with $8$ TBB threads per rank against a pure-MPI configuration using $192$ MPI ranks with one thread per rank. Both configurations utilize the same total of $192$ CPU cores, ensuring a fair comparison while isolating the effect of intra-process thread parallelism.

Figure~\ref{fig:count_ablation} reports the count-phase execution time across all datasets. The hybrid configuration consistently outperforms the pure-MPI configuration, achieving an average speedup of $20.2\times$. The largest improvements are observed on the smaller datasets, reaching up to $97\times$ on DBLP, where the computation is relatively small and the communication and synchronization overhead of managing $192$ MPI processes becomes more significant. As the dataset size increases, the counting workload dominates the execution time, reducing the relative impact of MPI overhead. Consequently, the speedup decreases to $1.9\times$ on NX and $2.1\times$ on YH. Overall, the results demonstrate that combining MPI with shared-memory parallelism significantly reduces communication overhead while maintaining efficient utilization of the available CPU cores.


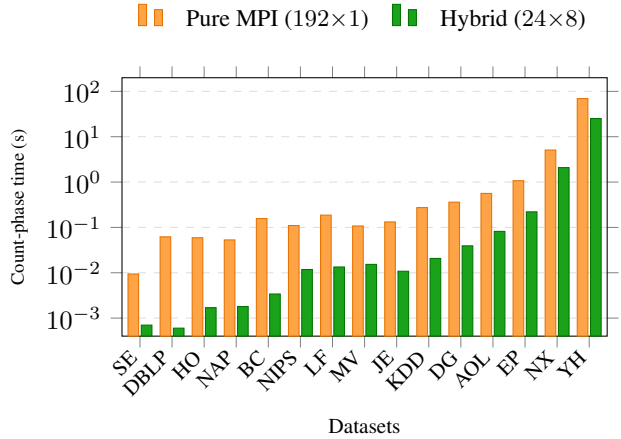
\begin{figure}[!h]

\begin{tikzpicture}
\begin{axis}[
    width=8.0cm, height=5cm,
    ybar=1pt, bar width=4pt,
    ymode=log, log basis y=10, log origin y=infty,
    symbolic x coords={SE,DBLP,HO,NAP,BC,NIPS,LF,MV,JE,KDD,DG,AOL,EP,NX,YH},
    xtick=data,
    xticklabel style={rotate=45, anchor=east, font=\footnotesize},
    xlabel={\footnotesize Datasets},
    ylabel={\scriptsize Count-phase time (s)},
    ymin=0.0004, ymax=200,
    ytick={0.001,0.01,0.1,1,10,100},
    yticklabels={$10^{-3}$,$10^{-2}$,$10^{-1}$,$10^{0}$,$10^{1}$,$10^{2}$},
    ymajorgrids=true, grid style={dashed, gray!25},
    tick align=outside,
    legend style={at={(0.5,1.14)}, anchor=south, legend columns=2,
                  draw=none, column sep=1.5ex, font=\small},
    enlarge x limits=0.04,
    every axis plot/.append style={fill opacity=0.9, draw opacity=1},
]
\addplot[fill=orange!80, draw=orange!90!black] coordinates {
(SE,0.0094)(DBLP,0.0620)(HO,0.0592)(NAP,0.0530)(BC,0.1569)(NIPS,0.1100)(LF,0.1866)(MV,0.1075)(JE,0.1320)(KDD,0.2735)(DG,0.3603)(AOL,0.5616)(EP,1.0713)(NX,5.0894)(YH,69.3599)
};
\addlegendentry{Pure MPI ($192{\times}1$)}

\addplot[fill=green!60!black, draw=green!45!black] coordinates {
(SE,0.0007)(DBLP,0.0006)(HO,0.0017)(NAP,0.0018)(BC,0.0034)(NIPS,0.0118)(LF,0.0134)(MV,0.0153)(JE,0.0108)(KDD,0.0207)(DG,0.0392)(AOL,0.0817)(EP,0.2204)(NX,2.0827)(YH,25.1964)
};
\addlegendentry{Hybrid ($24{\times}8$)}
\end{axis}
\end{tikzpicture}
\caption{Count-phase time of the hybrid configuration versus a pure-MPI.}
\label{fig:count_ablation}
\end{figure}

%


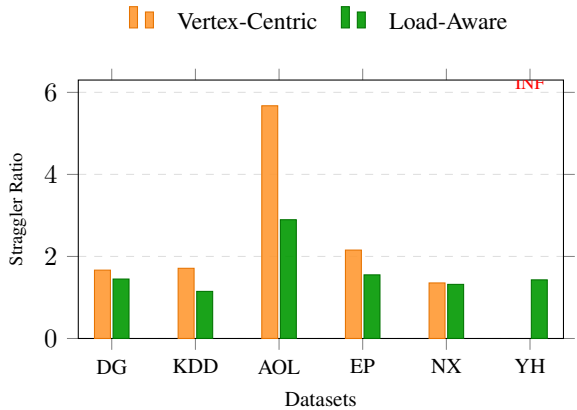
\begin{figure}[!h]
\centering
\begin{tikzpicture}
\begin{axis}[
    width=8.0cm,
    height=5cm,
    ybar=1pt,
    bar width=6pt,
    symbolic x coords={DG,KDD,AOL,EP,NX,YH},
    xtick={DG,KDD,AOL,EP,NX,YH},
    xticklabel style={font=\footnotesize},
    xlabel={\footnotesize Datasets},
    ylabel={\scriptsize Straggler Ratio},
    ymin=0,
    ymax=6.3,
    ymajorgrids=true,
    grid style={dashed, gray!25},
    tick align=outside,
    legend style={
        at={(0.5,1.14)},
        anchor=south,
        legend columns=2,
        draw=none,
        column sep=1.5ex,
        font=\small
    },
    enlarge x limits=0.08,
    every axis plot/.append style={
        fill opacity=0.9,
        draw opacity=1
    },
]

\addplot[
    fill=orange!80,
    draw=orange!90!black
] coordinates {
    (DG,1.666)
    (KDD,1.711)
    (AOL,5.673)
    (EP,2.155)
    (NX,1.354)
};
\addlegendentry{Vertex-Centric}

\addplot[
    fill=green!60!black,
    draw=green!45!black
] coordinates {
    (DG,1.449)
    (KDD,1.148)
    (AOL,2.895)
    (EP,1.551)
    (NX,1.319)
    (YH,1.429)
};
\addlegendentry{Load-Aware}

\node[
    font=\scriptsize,
    text=red,
    anchor=south
] at (axis cs:YH,5.85) {INF};

\end{axis}
\end{tikzpicture}

\caption{Comparison of the straggler ratio achieved by the vertex-centric and
load-aware partitioning strategies. Lower values indicate better workload balance;
INF indicates that the vertex-centric approach did not terminate within the
allowed execution time.}
\label{fig:straggler}
\end{figure}

\subsection{Effectiveness of Load-Aware Partitioning.}
Figure~\ref{fig:straggler} compares the workload imbalance achieved by the naive vertex-centric and the proposed load-aware partitioning strategies. The proposed method consistently reduces the straggler ratio across all datasets. The largest improvement is observed on AOL, where the straggler ratio decreases from 5.67 to 2.89 (48.97\%), followed by KDD (32.90\%), EP (28.04\%), and DG (13.04\%). In contrast, NX shows only a modest improvement (2.63\%) since its workload is already relatively balanced. For the largest dataset, YH, the naive partitioning fails to complete due to severe workload skew and memory exhaustion, whereas the proposed load-aware strategy successfully completes execution with a straggler ratio of 1.43. These results demonstrate that workload-aware partitioning effectively mitigates workload imbalance and improves the robustness of the proposed distributed algorithm.

NOTE:
The percentage improvement in workload balance is computed as

\begin{equation}
\label{eq:improvement}
\text{Improvement (\%)}=
\frac{\mathrm{SR}_{\text{naive}}-\mathrm{SR}_{\text{load-aware}}}
{\mathrm{SR}_{\text{naive}}}\times100,
\end{equation}

where $\mathrm{SR}_{\text{naive}}$ and $\mathrm{SR}_{\text{load-aware}}$ denote the straggler ratios obtained using the naive vertex-centric and the proposed load-aware partitioning strategies, respectively.
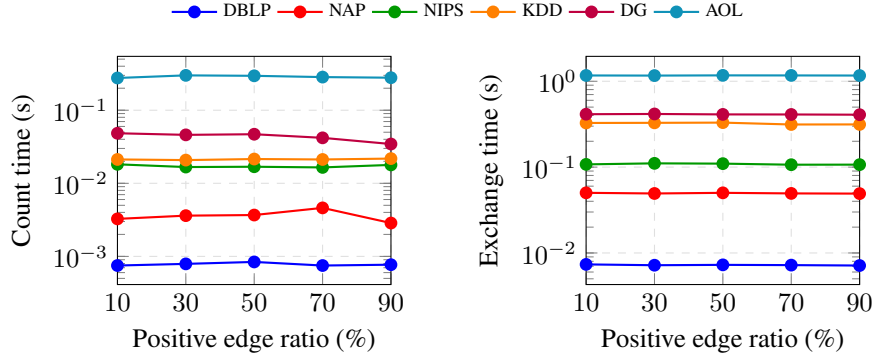
\begin{figure*}[ht]
\centering
\begin{tikzpicture}

\begin{axis}[
at={(0,0)},
anchor=south west,
width=5.2cm,
height=4.6cm,
xlabel={Positive edge ratio (\%)},
ylabel={Count time (s)},
ymode=log,
log basis y=10,
xmin=10,xmax=90,
xtick={10,30,50,70,90},
legend style={
draw=none,
font=\scriptsize,
at={(1.25,1.28)},
anchor=north,
legend columns=6},
grid=major,
grid style={dashed,gray!25},
every axis plot/.append style={thick,mark=*}
]

\addplot[color=blue]
coordinates{(10,0.00075)(30,0.00079)(50,0.00084)(70,0.00075)(90,0.00077)};
\addlegendentry{DBLP}

\addplot[color=red]
coordinates{(10,0.00326)(30,0.00361)(50,0.00368)(70,0.00461)(90,0.00286)};
\addlegendentry{NAP}

\addplot[color=green!60!black]
coordinates{(10,0.01819)(30,0.01669)(50,0.01685)(70,0.01650)(90,0.01786)};
\addlegendentry{NIPS}

\addplot[color=orange]
coordinates{(10,0.02121)(30,0.02078)(50,0.02148)(70,0.02119)(90,0.02179)};
\addlegendentry{KDD}

\addplot[color=purple]
coordinates{(10,0.04840)(30,0.04601)(50,0.04694)(70,0.04189)(90,0.03448)};
\addlegendentry{DG}

\addplot[color=cyan!70!black]
coordinates{(10,0.27614)(30,0.29979)(50,0.29564)(70,0.28405)(90,0.27851)};
\addlegendentry{AOL}

\end{axis}

\begin{axis}[
at={(6.2cm,0)},
anchor=south west,
width=5.2cm,
height=4.6cm,
xlabel={Positive edge ratio (\%)},
ylabel={Exchange time (s)},
ymode=log,
log basis y=10,
xmin=10,xmax=90,
xtick={10,30,50,70,90},
grid=major,
grid style={dashed,gray!25},
every axis plot/.append style={thick,mark=*}
]

\addplot[color=blue]
coordinates{(10,0.00736)(30,0.00720)(50,0.00727)(70,0.00722)(90,0.00714)};

\addplot[color=red]
coordinates{(10,0.05018)(30,0.04931)(50,0.05016)(70,0.04940)(90,0.04913)};

\addplot[color=green!60!black]
coordinates{(10,0.10761)(30,0.11063)(50,0.10978)(70,0.10667)(90,0.10685)};

\addplot[color=orange]
coordinates{(10,0.32679)(30,0.32789)(50,0.32986)(70,0.31366)(90,0.31425)};

\addplot[color=purple]
coordinates{(10,0.41333)(30,0.41589)(50,0.41112)(70,0.41034)(90,0.40784)};

\addplot[color=cyan!70!black]
coordinates{(10,1.16828)(30,1.16449)(50,1.17111)(70,1.16915)(90,1.16578)};

\end{axis}

\end{tikzpicture}
\caption{Impact of varying the positive-to-negative edge ratio on counting time and exchange time across all datasets.}
\label{fig:sign_ratio}
\end{figure*}

\subsection{Sign Distribution}
 To assess the robustness of the proposed framework under different edge-sign distributions, we vary the positive-to-negative edge ratio from 10\% to 90\% while preserving the underlying graph topology. Figure~\ref{fig:sign_ratio} shows that both the counting time and communication time remain nearly unchanged across all evaluated datasets, indicating that the computational cost of the proposed algorithm is largely independent of the edge-sign distribution. In contrast, the number of balanced butterflies varies substantially with the sign ratio, as expected, because balanced butterfly formation is directly determined by the arrangement of positive and negative edge labels. Overall, these results demonstrate that the proposed framework delivers stable runtime performance across diverse sign distributions while accurately reflecting the corresponding changes in balanced butterfly counts.

\subsection{Symmetric and Asymmetric Wedges}

To evaluate the effectiveness of separating symmetric and asymmetric wedges into distinct buckets, we compare the proposed approach with a baseline that does not employ this separation. Without bucket-based separation, the algorithm must first enumerate candidate wedge pairs and then explicitly verify whether each pair forms a balanced butterfly by checking the signs of the four edges in the resulting $4$-cycle. This additional enumeration and verification introduces substantial computational overhead, as many candidate $4$-cycles must be examined and subsequently discarded because they correspond to unbalanced butterflies. The results presented in Table~\ref{tab:sym_asym_ablation} demonstrate the performance advantage of explicitly separating symmetric and asymmetric wedges into distinct buckets.

\begin{table}[t]
\centering
\caption{Study of symmetric-asymmetric wedge separation.}
\label{tab:sym_asym_ablation}
\begin{tabular}{|l|c|c|c|}
\hline
Dataset & Without Wedge (s) & With Separate Wedge (s) & Speedup \\
\hline
SE & 3.69 & 0.0.031628 & $\textbf{123}\times$ \\ \hline
HO & 157.00 & 0.935430 & $\textbf{128}\times$ \\
\hline
\end{tabular}
\end{table}

\section{Concluding Remarks}
\label{Sec:Con}
We studied balanced butterfly counting in signed bipartite graphs and presented \texttt{D-BBC}, a distributed algorithm based on a hybrid MPI+TBB model with a workload-aware partitioning strategy. Experiments on $15$ real-world datasets show that \texttt{D-BBC} consistently outperforms state-of-the-art serial, shared-memory, and distributed (\texttt{S-Monarch}) baselines, and scales to large graphs on which \texttt{S-Monarch} runs out of memory. The phase-wise analysis confirms that the counting phase scales efficiently with parallelism, while communication gradually becomes the dominant cost in multi-node execution.
As future work, we plan to reduce boundary-wedge communication through locality- and communication-aware partitioning, and to extend \texttt{D-BBC} to web-scale graphs and to temporal, dynamic, and higher-order motif counting.

\bibliographystyle{ieeetr}
\bibliography{references1}

\begin{IEEEbiography}[{\includegraphics[width=1in,height=1.25in,clip,keepaspectratio]{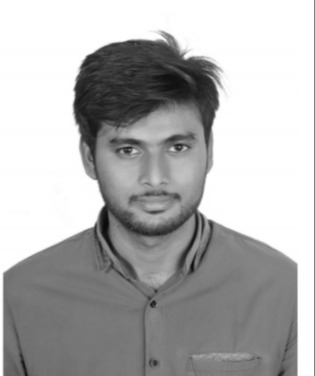}}]{\textbf{Mekala Kiran}} was born in Khammam, Telangana, India, in 
1994. He received the B.Tech. degree in Computer Science and Engineering from Jawaharlal Nehru Technological University (JNTU), Hyderabad, India, in 2016, and the M.Tech. degree in Computer Science and Engineering from Jawaharlal Nehru Technological University (JNTU), Hyderabad, India, in 2019. He is currently pursuing a Ph.D. degree in Computer Science and Information Systems at Birla Institute of Technology and Science (BITS) Pilani, Hyderabad Campus, Hyderabad, India.

His research interests include graph algorithms and high-performance computing. 
\end{IEEEbiography}

\begin{IEEEbiography}[{\includegraphics[width=1in,height=1.25in,clip,keepaspectratio]{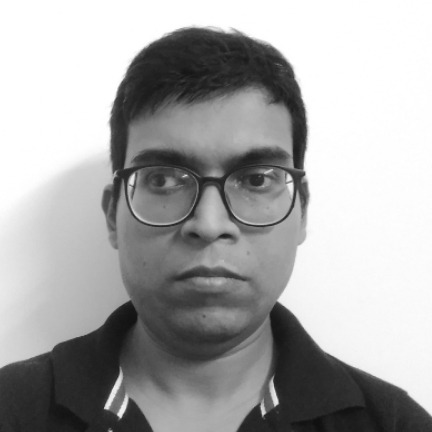}}]{\textbf{Apurba Das}}  is an Assistant Professor in the Department of Computer Science and Information Systems at Birla Institute of Technology and Science, Hyderabad Campus. He
earned his PhD in Computer Engineering from Iowa State University, USA, in 2019, and M.Tech in Computer Science from the Indian Statistical Institute, Kolkata. He completed a postdoctoral fellowship at the School of Computing, National University of Singapore, from 2019 to 2020. 

His research interests include large-scale graph analysis, parallel algorithms, data stream mining, and related areas.
\end{IEEEbiography}

\begin{IEEEbiography}[{\includegraphics[width=1in,height=1.25in,clip,keepaspectratio]{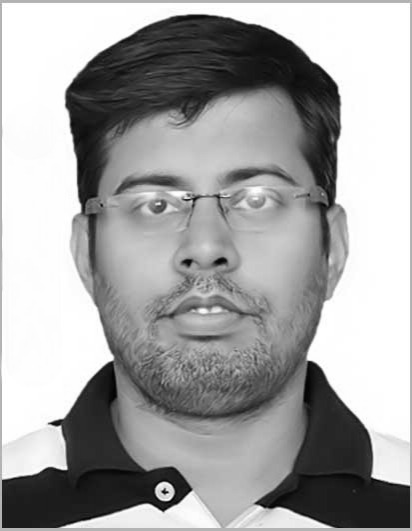}}]{\textbf{Suman Banerjee}}  obtained his Ph.D. from the Indian Institute of Technology Kharagpur
in 2020. After a short stay at the Indian Institute of Technology Gandhinagar as a Post
Doctoral Fellow, he joined the Department of Computer Science and Engineering, Indian
Institute of Technology Jammu, as an assistant professor in the same year. 

His research
interests include algorithm design with a particular focus on social and information
network analysis, scheduling and schedulability analysis in cloud computing systems, distributed and multi-core systems, real-time systems, etc.; structural pattern analysis for time-varying graphs; graph theory and graph algorithms; and parameterised complexity. He has published more than 50 research papers in international journals
and conferences.
\end{IEEEbiography}

\EOD

\end{document}